\documentclass[12pt]{amsart}
\usepackage{amssymb,filecontents}
\usepackage{a4wide}
\usepackage{url}
\usepackage{tikz}
\usepackage{times}
\usepackage{enumerate}

\def\alphabet{\Sigma}
\def\poly{\mathop{\mathrm{poly}}}

 \def\Nesetril{Ne{\v{s}}et{\v{r}}il}
\def\SAT{\textsc{\#Sat}}
\def\nCSP{\textsc{\#CSP}}
\def\BIS{\textsc{\#BIS}} 
\def\nHom#1{\textsc{\#HomsTo}(#1)}
\def\wHom#1{\textsc{\#WHomsTo}(#1)}
\def\MultiCutCount#1{\textsc{\#MultiterminalCut(#1)}}
\def\Potts{\textsc{Potts}}
\def\hPotts{\textsc{HyperPotts}}
\def\uhPotts{\textsc{UniformHyperPotts}}
\def\WE#1#2{\textsc{WE}(#1,#2)}
\def\bqcol#1{\textsc{\#Bipartite }#1\textsc{-Col}}

\def\JS{J_3^*}

\def\bichrom#1{\text{bi}(#1)}
\def\components#1{\kappa(#1)}
\def\APeq{\equiv_\mathrm{AP}}
\def\APred{\leq_\mathrm{AP}} 
\def\hypergraph{\mathcal{H}}
\def\hypervertices{\calV}
\def\hyperedges{\calE}
\def\calE{\mathcal{E}} 
\def\calV{\mathcal{V}}
\def\hyperedge{f} 
\def\ZTutte{Z_\mathrm{Tutte}}
\def\ZPotts{Z_\mathrm{Potts}}
\def\graph{G} 
\def\graphvertices{V}
\def\graphedges{E}
\def\hypervertex{v}

\def\nonnegQ{\mathbb{Q}_{\geq 0}}
\def\posQ{\mathbb{Q}_{> 0}}
 
\def\Hom#1#2{\mathop{\mathrm{Hom}}(#1,#2)}

\def\IMP{\mathrm{IMP}}

\def\FP{\mathrm{FP}}
\def\numP{\mathrm{\#P}}
\def\RHPi{\text{\#RH$\Pi_1$}}

\def\mono{\mathop{\mathrm{mono}}}
\def\sigmahat{\hat\sigma}
\def\Fp{\mathbb{F}_p}
\def\Fq{\mathbb{F}_q}
\let\phi=\varphi
\def\sat{\mathop{\mathrm{sat}}}
\def\unsat{\mathop{\mathrm{unsat}}}
\def\bfa{\mathbf a}
\def\bfb{\mathbf b}
\def\bsigma{\boldsymbol{\hat\sigma}}

\newtheorem{theorem}{Theorem}
\newtheorem{lemma}[theorem]{Lemma}
\newtheorem{corollary}[theorem]{Corollary}
\newtheorem{observation}[theorem]{Observation}

\title [Tree Homomorphisms]{The Complexity of Approximately Counting Tree Homomorphisms}\thanks{
The research leading to these results has received funding from the European Research Council under 
  the European Union's Seventh Framework Programme (FP7/2007-2013) ERC grant agreement no.\ 334828. The paper 
  reflects only the authors' views and not the views of the ERC or the European Commission. 
  The European Union is not liable for any use that may be made of the information contained therein.
  This work was partially supported by the EPSRC grant 
{\it Computational Counting.} 
  }

\author{Leslie Ann Goldberg}
\address{Leslie Ann Goldberg, Department of Computer Science,
University of Oxford,
Wolfson Bldg, Parks Rd.,
Oxford OX1 3QD, United Kingdom.}
 
\author{Mark Jerrum}
\address{Mark Jerrum, School of Mathematical Sciences\\
Queen Mary, University of London, Mile End Road, London E1 4NS, United Kingdom.}

\begin{document}
\maketitle
  
  \begin{abstract}  
  
We study two computational problems, parameterised by a fixed tree~$H$.
$\nHom{H}$ is the problem of counting homomorphisms from an input graph~$G$
to~$H$. $\wHom{H}$ is the problem of counting weighted homomorphisms to~$H$,
given an input graph~$G$ and a weight function for each vertex~$v$ of~$G$.
Even though $H$ is a tree, these problems turn out to be sufficiently rich to capture all of the 
known approximation behaviour in $\numP$.
We give a complete trichotomy for $\wHom{H}$. If $H$  is a star then $\wHom{H}$ is in $\FP$. 
If $H$ is not a star but it does not contain a certain induced subgraph~$J_3$ then
$\wHom{H}$ is equivalent under approximation-preserving (AP) reductions to $\BIS$, the
problem of counting independent sets in a bipartite graph. This problem is complete
for the class $\RHPi$ under AP-reductions. Finally, if $H$ contains an induced~$J_3$
then $\wHom{H}$ is equivalent under  AP-reductions to $\SAT$,
the problem of counting satisfying assignments to a CNF Boolean formula.
Thus, $\wHom{H}$ is complete for $\numP$ under AP-reductions.
The results are similar for $\nHom{H}$ except that a rich structure emerges if $H$
contains an induced~$J_3$.  We show that there are trees~$H$
for which $\nHom{H}$ is $\SAT$-equivalent (disproving a plausible conjecture of Kelk).
However, it is still not known whether $\nHom{H}$ is $\SAT$-hard for \emph{every}
tree~$H$ which contains an induced~$J_3$.
 It turns out that there is  an interesting connection between these homomorphism-counting
problems and the problem of approximating the partition function of the \emph{ferromagnetic
Potts model}.
In particular,  we show that for 
a family of graphs $J_q$, parameterised by a positive integer~$q$,
the problem $\nHom{J_q}$ is AP-interreducible with the problem
of approximating the partition function of the $q$-state Potts model.
It was not previously known that the Potts model had
a homomorphism-counting interpretation.
We use this connection to obtain some additional upper bounds for 
the approximation complexity of $\nHom{J_q}$. 
   \end{abstract}
   
\section{Introduction}   
A \emph{homomorphism} from a graph~$G$ to a graph~$H$ is
a mapping $\sigma:V(G)\rightarrow V(H)$
such that
the image $(\sigma(u),\sigma(v))$  
of every edge $(u,v) \in E(G)$  is in $E(H)$.
Let $\Hom GH$ denote the set of homomorphisms from~$G$ to~$H$
and let 
$Z_H(G)=|\Hom GH|$.
For each fixed~$H$, we consider the following computational problem.
\begin{description}
\item[Problem] $\nHom{H}$.
\item[Instance] Graph  $G$.
\item[Output]  $Z_H(G)$.
\end{description} 
The vertices of~$H$ are often referred to as ``colours''
and a homomorphism from~$G$ to~$H$ can be thought of as 
an assignment of colours to the vertices of $G$ 
which satisfies
certain constraints along each edge of~$G$.
The constraints guarantee that adjacent vertices in~$G$
are assigned colours which are adjacent in~$H$. 
A homomorphism in $\Hom GH$ is therefore 
often called an ``$H$-colouring'' of~$G$. When $H=K_q$, 
the complete graph with $q$~vertices, the elements of $\Hom G{K_q}$  
are 
proper $q$-colourings of~$G$.

There has been much work on determining the complexity
of the $H$-colouring decision problem,
which is the problem of 
determining whether $Z_H(G)=0$, given input~$G$. 
This work will be described in Section~\ref{sec:prev}, but at this point it is worth mentioning
the dichotomy result of Hell and \Nesetril~\cite{HN}, which shows that the decision problem is 
solvable in polynomial time if $H$ is bipartite and that it is NP-hard otherwise.
There has also been work~\cite{DG,Kelk}  on determining the complexity of 
exactly or approximately solving the related counting problem~$\nHom{H}$.
This paper is concerned with the computational difficulty of  $\nHom{H}$ when $H$ is
bipartite, and particularly when $H$ is a tree.  
 
As an example, consider the case where $H$ is the four-vertex path~$P_4$ (of length three).  
Label the vertices (or colours) $1,2,3,4$, in sequence.  If $G$ is not bipartite then $\Hom GH=\emptyset$,
so the interesting case is when $G$ is bipartite.  Suppose for simplicity that $G$ is connected.
Then one side of the vertex bipartition of $G$ must be assigned even colours and the other side 
must be assigned odd colours.
It is easy to see that the vertices assigned colours~$1$ and~$4$ form an independent set of $G$, and that
every independent set arises in exactly two ways as a homomorphism.
Thus, $Z_{P_4}(G)$ is equal to twice the number of 
independent sets in the bipartite graph~$G$.  We will return to this example presently.
 
It will sometimes be useful to consider 
a weighted generalisation of the homomorphism-counting problem.
Suppose, for each $v\in V(G)$, that
$w_v:V(H)\rightarrow  \nonnegQ$ is a weight function, assigning a non-negative
rational weight to each colour. Let $W(G,H)$
be 
an indexed set
of weight functions, containing one weight function for each vertex $v\in V(G)$,
 Thus, $$W(G,H)
=\{ w_v \mid v\in V(G)\}.$$
 Our goal is to compute the weighted sum of homomorphisms from~$G$
 to~$H$,
 which is expressed as the partition function
  $$Z_{H}(G, W(G,H)) = \sum_{\sigma\in \Hom GH} 
\prod_{v\in V(G)} w_v(\sigma(v)).$$ 

Given a fixed~$H$,
each weight function $w_v\in W(G,H)$ can
be represented succinctly as a list of $|V(H)|$ 
rational numbers.
This representation is used in the following computational problem.
\begin{description}
\item[Problem] $\wHom{H}$.
\item[Instance] A graph  $G$ and 
an indexed set
of weight functions $W(G,H)$.
\item[Output]  $Z_{H}(G,W(G,H))$.
\end{description} 

The complexity of \emph{exactly} solving $\nHom{H}$ and $\wHom{H}$ 
is already understood.
Dyer and Greenhill have observed \cite[Lemma 4.1]{DG} that $\nHom{H}$
is in~$\FP$ if $H$ is a complete bipartite graph.  
It is easy to see (see Observation~\ref{obs:star})  
that the same is true of $\wHom{H}$.
On the other hand, Dyer and Greenhill showed that $\nHom{H}$ is $\numP$-complete
for every bipartite graph~$H$ that is not complete.
Since $\nHom{H}$ is a special case of 
the more general problem 
$\wHom{H}$, we conclude that
both problems are in~$\FP$ if $H$ is a star (a  tree in which some ``centre'' vertex is 
an endpoint of every edge),
and that both problems are $\numP$-complete for every other tree~$H$.

This paper maps the complexity of  \emph{approximately} solving $\nHom{H}$
and $\wHom{H}$ when $H$ is a tree.  
Dyer, Goldberg, Greenhill and Jerrum~\cite{APred} introduced the 
concept of ``AP-reduction'' for studying the complexity of approximate counting
problems.  Informally, an AP-reduction is an efficient 
reduction from one counting problem to another, which preserves closeness of 
approximation;  two counting problems that are interreducible using this kind
of reduction have the same complexity when it comes to finding good approximate 
solutions.  We have already encountered an extremely simple example of two 
AP-interreducible problems, namely $\nHom{P_4}$ and $\BIS$, the problem 
of counting independent sets in a bipartite graph.  Using  less trivial 
reductions, Dyer et al.\ showed (\cite[Theorem 5]{APred}) that several natural
counting problems in addition to $\nHom{P_4}$ 
are interreducible with $\BIS$, and moreover that they are all
complete for the complexity class $\RHPi$ 
with respect to AP-reductions.  The class $\RHPi$ is conjectured to contain problems that 
do not  have an FPRAS;  however it is not believed to contain $\SAT$, the classical 
hard problem of computing the number of satisfying assignments to a CNF Boolean formula.
Refer to Section~\ref{sec:prelim} for more detail on the technical concepts mentioned here 
and elsewhere in the introduction.

Steven Kelk's PhD thesis~\cite{Kelk} 
examined the approximation complexity of the problem $\nHom{H}$
for general~$H$.
He identified certain families of graphs~$H$
for which $\nHom{H}$ is AP-interreducible with $\BIS$ and other
large families for which $\nHom{H}$ is AP-interreducible with $\SAT$.
He noted \cite[Section 5.7.1]{Kelk} that, during the study, he did not
encounter \emph{any} bipartite graphs~$H$ for which $\SAT \APred \nHom{H}$,
and that he suspected \cite[Section 7.3]{Kelk} that there were ``structural barriers''
which would prevent  homomorphism-counting problems 
to bipartite graphs
from being $\SAT$-hard.
An interesting test case is the tree $J_3$ which is depicted in Figure~\ref{fig:J3}.
Kelk referred to this tree \cite[Section 7.4]{Kelk} as ``the junction'', and conjectured that
$\nHom{J_3}$ is neither $\BIS$-easy nor $\SAT$-hard.
Thus, he conjectured that unlike the setting of Boolean constraint satisfaction,
where every parameter leads to a computational problem 
which is FPRASable, $\BIS$-equivalent, or $\SAT$-equivalent \cite{trichotomy},
the  complexity landscape for approximate $H$-colouring may be more nuanced, in the sense
that there might be graphs~$H$ for which none of these hold.

The purpose of this paper is to describe the interesting complexity landscape of
the approximation problems $\nHom{H}$ and $\wHom{H}$ when $H$ is a tree.
It turns out that even the case in which $H$ is a tree is sufficiently rich to
include all of the known approximation complexity behaviour in $\numP$.

First, consider the weighted problem~$\wHom{H}$.
For this problem, we show that there is a complexity trichotomy, and
the trichotomy depends upon the induced subgraphs of~$H$.
We say that $H$ {\it contains an induced $H'$} if $H$ has an induced
subgraph that is isomorphic to~$H'$.
Here is the result.
If $H$ contains no induced $P_4$ then it is a star, 
so $\wHom{H}$ is in $\FP$ (Observation~\ref{obs:star}).
If $H$ contains an induced $P_4$ but it does not contain an induced $J_3$
then it turns out that $\wHom{H}$ is AP-interreducible with $\BIS$ (Lemma~\ref{lem:intermediate}).
Finally, if $H$ contains an induced $J_3$,
then $\SAT \APred \wHom{H}$ (Lemma~\ref{lem:hardweighted}.)
Thus, the complexity of $\wHom{H}$ is completely determined by the induced subgraphs of
the tree~$H$, and 
there
are no possibilities other than those that arise in the 
Boolean constraint satisfaction trichotomy~\cite{trichotomy}.

Now consider the problem~$\nHom{H}$.  
Like its weighted counterpart, the unweighted problem $\nHom{H}$ is in $\FP$ if $H$ is a star,
and it is $\BIS$-equivalent if $H$ contains an induced~$P_4$ but it does not contain an induced~$J_3$.
However, it is   not known whether $\nHom{H}$ is $\SAT$-hard
for every $H$ which contains an induced $J_3$.
The structure that has emerged is already quite rich.
First, we have discovered  (Theorem~\ref{thm:hardH}) that there are trees $H$
for which $\nHom{H}$ is $\SAT$-hard.
This result is surprising --- it disproves the plausible conjecture of Kelk
that $\nHom{H}$ is not $\SAT$-hard for any bipartite graph~$H$.
We don't know whether $\nHom{H}$ is $\SAT$-hard for \emph{every}
tree~$H$ which contains an induced~$J_3$.
In fact, we have discovered an interesting connection between these homomorphism-counting
problems and the problem of approximating the partition function of the \emph{ferromagnetic
Potts model}.
In particular, Theorem~\ref{thm:junction} shows that for 
a family of graphs $J_q$, parameterised by a positive integer~$q$,
the problem $\nHom{J_q}$ is AP-interreducible with the problem
of approximating the partition function of the $q$-state Potts model.
This is surprising because it was not known that the Potts model had
a homomorphism-counting interpretation.
 
The  
Potts-model connection allows us to give a non-trivial upper bound
for the complexity of $\nHom{J_q}$. In particular, Corollary~\ref{cor:bqcol}
shows that this 
problem
is AP-reducible to the problem of counting proper $q$-colourings of bipartite graphs.
 
We are not aware of any complexity relationships between the problems $\nHom{J_q}$,
for $q>2$.  At one extreme, they might all be AP-interreducible;  at the other, they 
might all be incomparable.  Another conceivable situation is that \nHom{$J_q$} is 
AP-reducible to $\nHom{J_{q'}}$ exactly when $q\leq q'$.  There is no real evidence for
or against any of these or other possibilities.  However, in the final section we exhibit
a natural problem that provides an upper bound on the complexity of infinite families
of problems of the form $\nHom{J_q}$ where $q$ is a prime power.  Specifically, we
show (Corollary~\ref{newcor}) that $\nHom{J_{p^k}}$ is AP-reducible 
to the weight enumerator of a linear code over the field~$\Fp$.
 
\subsection{Previous Work} 
\label{sec:prev}

We have already mentioned Hell and \Nesetril's classic work~\cite{HN} on the 
complexity of the
$H$-colouring decision problem. 
They showed that this problem is solvable in polynomial time if $H$ is bipartite, and
that it is NP-complete otherwise. 
Our paper is concerned with the situation in which $H$ is an undirected graph
(specifically, an undirected tree) but it is worth noting that the decision problem becomes much
more complicated if $H$ is allowed to be a \emph{directed} graph.
Indeed, Feder and Vardi showed~\cite{FV} that every constraint satisfaction problem (CSP)
is equivalent to some digraph homomorphism problem.
Despite much research, a complete dichotomy theorem for 
the digraph homomorphism decision problem is not known.
Bang-Jensen and Hell~\cite{BJH} had conjectured a dichotomy for
the special case  in which the digraph~$H$ has
no sources and no sinks. This conjecture was proved in
important recent work of Barto, Kozik and Niven~\cite{BKN}.
Given the conjecture,
Hell, \Nesetril, and Zhu~\cite{HNZ} stated that
``digraphs with sources and sinks, and in particular oriented trees, seem to be the hard part of the problem.''
Gutjahr, Woeginger and Welzl~\cite{GWW} constructed a directed tree~$H$
such that determining whether a digraph~$G$ has a homomorphism to~$H$ is NP-complete.
Of course, for 
some
other trees, this problem is solvable in polynomial time.
For example, they showed that it is solvable in polynomial time whenever $H$ is an oriented path
(a path in which edges may go in either direction). 
Hell, \Nesetril\ and Zhu~\cite{HNZ} construct a whole family of directed trees for which the
homomorphism decision problem is
NP-hard, and study the problem of characterising NP-hard trees by forbidden subtrees.
The reader is referred to Hell and \Nesetril's book~\cite{HNbook} and to their survey
paper~\cite{HNsurvey} for more details about these decision problems.

As mentioned in the introduction, there is already some existing work~\cite{DG, Kelk} on determining the complexity of exactly or approximately counting homomorphisms. This work is discussed in more detail elsewhere
in this paper.
The problem of  sampling homomorphisms uniformly at random
(or, in the weighed case, of sampling homomorphisms with probability proportional
to their contributions to the partition function) is closely related to the approximate
counting problem. We will later discuss some existing work~\cite{GKP} on the complexity of the
homomorphism-sampling problem.
First, we describe some related results on a particular approach to this problem - namely,  the application of
the Markov chain Monte Carlo (MCMC) method.
Here the idea is to simulate a Markov chain whose states correspond to 
homomorphisms from~$G$ to~$H$.
The chain will be constructed so that the probability of a particular homomorphism~$\sigma$
in the stationary distribution of the chain is proportional to the contribution of~$\sigma$ to the
partition function. If the Markov chain is \emph{rapidly mixing} then it is possible to
efficiently sample homomorphisms from a distribution that is very close to the appropriate distribution.
This, in turn, leads to a good approximate counting algorithm~\cite{HColSampleCount}. 
First, Cooper, Dyer and Frieze~\cite{CDF}
considered the unweighted problem. They showed that, for any non-trivial $H$,
any Markov chain on $H$-colourings that
changes the colours of up to some constant fraction of the vertices of~$G$ in a single step
will have exponential mixing time (so will not lead to an efficient approximate counting algorithm).
When $H$ is a tree with a self-loop on every vertex, they construct
a weight function $w_H\colon V(H) \to \nonnegQ$  
so that rapid mixing does occur for the special case of 
the weighted homomorphism problem in which every vertex $v$ of~$G$ has weight function $w_v=w_H$.
Thus,  their result gives an FPRAS for this special case of $\wHom{H}$.
The slow-mixing results of \cite{CDF} have been extended in~\cite{BS}
and in \cite{BCDT}. In particular, Borgs et al.~\cite{BCDT}
considered the case in which $H$ is 
a rectangular subset of the hypercubic lattice, and
constructed a weight function~$w_H$ for which  
quasi-local Markov chains (which change the colours of up to some constant fraction of the
vertices in a small sublattice at each step)
have slow mixing.

\section{Preliminaries}
\label{sec:prelim}

This section brings
together the main complexity-theoretic notions that 
are specific to the study of approximate counting problems.  A more detailed
account can be found in~\cite{APred}.

A \emph{randomised approximation scheme\/} is an algorithm for
approximately computing the value of a function~$f:\Sigma^*\rightarrow
\mathbb{R}_{\geq 0}$.
The approximation scheme has a parameter~$\varepsilon\in(0,1)$ which specifies
the error tolerance.
A \emph{randomised approximation scheme\/} for~$f$ is a
randomised algorithm that takes as input an instance $ x\in
\alphabet^{\ast }$ (e.g., in the case of $\nHom{H}$, the input would 
be an encoding of a graph~$G$) and a rational error
tolerance $\varepsilon \in(0,1)$, and outputs a rational number $z$
(a random variable depending on the ``coin tosses'' made by the algorithm)
such that, for every instance~$x$,
$\Pr \big[e^{-\epsilon} f(x)\leq z \leq e^\epsilon f(x)\big]\geq \tfrac{3}{4}$.
We adopt the convention that~$z$ is represented as a pair of integers representing the numerator
and the denominator.
The randomised approximation scheme is said to be a
\emph{fully polynomial randomised approximation scheme},
or \emph{FPRAS},
if it runs in time bounded by a polynomial
in $ |x| $ and $ \epsilon^{-1} $.
As in~\cite{FerroPotts}, we say that a real number~$z$ is 
\emph{efficiently approximable} if there is an FPRAS for the constant function $f(x)=z$.

Our main tool for understanding the relative difficulty of
approximation counting problems is \emph{approximation-preserving reductions}.
We use the notion of 
approximation-preserving reduction from Dyer et al.~\cite{APred}.
Suppose that $f$ and $g$ are functions from
$\alphabet^{\ast }$ to~$\mathbb{R}_{\geq 0}$. An  AP-reduction 
from~$f$ to~$g$ gives a way to turn an FPRAS for~$g$
into an FPRAS for~$f$. 
The actual definition in~\cite{APred} applies to functions whose outputs are natural numbers.
The generalisation that we use here follows McQuillan~\cite{McQuillan}.
An {\it approximation-preserving reduction\/} (AP-reduction)
from $f$ to~$g$ is a randomised algorithm~$\mathcal{A}$ for
computing~$f$ using an oracle for~$g$. The algorithm~$\mathcal{A}$ takes
as input a pair $(x,\varepsilon)\in\alphabet^*\times(0,1)$, and
satisfies the following three conditions: (i)~every oracle call made
by~$\mathcal{A}$ is of the form $(w,\delta)$, where
$w\in\alphabet^*$ is an instance of~$g$, and $\delta \in (0,1)$ is an
error bound satisfying $\delta^{-1}\leq\poly(|x|,
\varepsilon^{-1})$; (ii) the algorithm~$\mathcal{A}$ meets the
specification for being a randomised approximation scheme for~$f$
(as described above) whenever the oracle meets the specification for
being a randomised approximation scheme for~$g$; and (iii)~the
run-time of~$\mathcal{A}$ is polynomial in $|x|$ and
$\varepsilon^{-1}$ and the bit-size of the values returned by the oracle.

If an approximation-preserving reduction from $f$ to~$g$
exists we write $f\APred g$, and say that {\it $f$ is AP-reducible
  to~$g$}.
Note that if $f\APred g$ and $g$ has an FPRAS then $f$ has an FPRAS\null.
(The definition of AP-reduction was chosen to make this true.)
If $f\APred g$ and $g\APred f$ then we say that
{\it $f$ and $g$ are AP-interreducible}, and write $f\APeq g$.
A word of warning about terminology:
the notation $\APred$ has been
used 
(see, e.g., \cite{CrescenziGuide})
to denote a different type of approximation-preserving reduction which applies to
optimisation problems.
We will not study optimisation problems in this paper, so hopefully this will
not cause confusion.

Dyer et al.~\cite{APred} studied counting problems in \#P and
identified three classes of counting problems that are interreducible
under approx\-imation-preserving reductions. The first class, containing the
problems that have an FPRAS, are trivially AP-interreducible since
all the work can be embedded into the reduction (which declines to
use the oracle). The second class  is the set of problems that are
AP-interreducible with \SAT, the problem of counting
satisfying assignments to a Boolean formula in CNF\null.
Zuckerman~\cite{zuckerman}
has shown that \SAT{} cannot have an FPRAS unless
$\mathrm{RP}=\mathrm{NP}$. The same is obviously true of any problem
 to which \SAT{} is AP-reducible.  
 
The third class appears to be of intermediate complexity.
It contains   all of the counting problems
expressible in a certain logically-defined complexity class, $\RHPi$. Typical
complete problems include counting the downsets in a partially ordered
set~\cite{APred},
computing the partition function of the ferromagnetic Ising model with
local external magnetic fields~\cite{Ising},
and counting the independent sets in a bipartite graph,
which is defined as follows.

\begin{description}
\item[Problem] $\BIS$.
\item[Instance] A bipartite graph $G$.
\item[Output]  The number of independent sets in $G$.
\end{description}

In \cite{APred} it was shown that $\BIS$ is complete for the
logically-defined
complexity class $\mathrm{\#RH}\Pi_1$  with respect to approximation-preserving
reductions.
We conjecture~\cite{FerroPotts} 
that there is no FPRAS for $\BIS$.

A problem 
that is closely related to approximate counting
is the problem of sampling configurations 
almost uniformly at random.
The analogue of an FPRAS in the context of sampling problems is the PAUS, or
\emph{Polynomial Almost Uniform Sampler}.

Goldberg, Kelk, and Paterson~\cite{GKP} have studied the problem of sampling $H$-colourings
almost uniformly at random.  They gave a hardness result for every fixed tree~$H$ that
is not a star. In particular, their theorem \cite[Theorem 2]{GKP} 
shows that there is no PAUS for sampling $H$-colourings unless $\BIS$ has an FPRAS.

In general, there is a close connection between approximate counting and
almost-uniform sampling. Indeed, in the presence of a technical 
condition called ``self-reducibility'', the counting and sampling variants 
of two problems
are interreducible~\cite{JVV}.  
The weighted problem $\wHom{H}$ is self-reducible,
so the result of~\cite{GKP} immediately gives an AP-reduction from~$\BIS$
to $\wHom{H}$ for every tree~$H$ that is not a star.
However, it is not known whether the unweighted problem~$\nHom{H}$ is
self-reducible.

As mentioned in Section~\ref{sec:prev} the paper \cite{HColSampleCount}
shows how to turn a PAUS for $H$-colourings into an FPRAS for $\nHom{H}$,
but it is not known whether there is a reduction in the other direction.
Thus, we cannot directly apply the hardness result of~\cite{GKP} to
reduce~$\BIS$ to~$\nHom{H}$. 
However, we will see in the next section that the complexity gap
between problems with an FPRAS and those that are $\BIS$-equivalent
still holds
for $\nHom{H}$ in the special case when $H$ is a tree,
which is the focus of this paper.
 
\section{Weighted tree homomorphisms}
\label{sec:weighted}

First, we introduce some notation and a few graphs that are of special interest.

In this paper, the graphs that we consider are undirected 
and simple --- they do not have self-loops or multiple edges between vertices.
For every positive integer~$n$, let $[n]$ denote $\{1,2,\ldots,n\}$.  
We use $\Gamma_H(v)$ to denote the set of neighbours of vertex~$v$ in graph~$H$
and we use $d_H(v)$ to denote 
the degree of~$v$, which is~$|\Gamma_H(v)|$. 

Let $P_n$ be the $n$-vertex path (with $n-1$ edges).
An $n$-leaf \emph{star} is the complete bipartite graph
$K_{1,n}$.
Let $J_q$ be the graph with vertex set 
$$V(J_q) =  \{w\} \cup \{c_i\mid i\in[q]\} \cup \{c'_i\mid i\in[q]\},$$
and edge set 
$$ E(J_q) = \{(c_i,c'_i) \mid i\in[q]\} \cup \{(c'_i,w)\mid i\in[q]\}.$$
$J_3$ is depicted in Figure~\ref{fig:J3}.
\begin{figure}
\centering{
 \begin{tikzpicture}[fill=white,scale=0.7,
line width=0.5pt,inner sep=1pt,minimum size=2.5mm]
\pgfsetxvec{\pgfpoint{1.7cm}{0cm}}
\pgfsetyvec{\pgfpoint{0cm}{1.7cm}}
\path 
(0,0) node [draw,fill,circle,minimum size=0.6cm](w){ $w$}
(1,0) node [draw,fill,circle,minimum size=0.6cm](y0){ $c'_2$}
(-1,0) node [draw,fill,circle,minimum size=0.6cm](x0){ $c'_1$}
(0,-1) node [draw,fill,circle,minimum size=0.6cm](z0){ $c'_3$}
(-2,0) node [draw,fill,circle,minimum size=0.6cm](x1){ $c_1$}
(2,0) node [draw,fill,circle,minimum size=0.6cm](y1){ $c_2$}
(0,-2) node [draw,fill,circle,minimum size=0.6cm](z1){ $c_3$}
;
\draw [-] (x1)--(x0) -- (w) -- (y0) -- (y1);
\draw [-] (w)--(z0)--(z1); 

\end{tikzpicture}

 }
 \caption{The tree $J_3$.}
\label{fig:J3}
\end{figure}
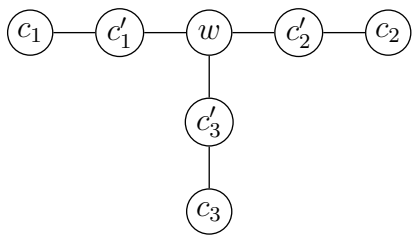

\subsection{Stars}

As 
Dyer and Greenhill observed
\cite[Lemma 4.1]{DG}, $\nHom{H}$ is in $\FP$ if
$H$ is a complete bipartite graph.
We now show 
that $\wHom{H}$ is also in $\FP$ in this case.
Suppose that $H$ is a complete bipartite graph with 
bipartition $(U,U')$ where $U=\{u_1,\ldots,u_h\}$ and
$U'=\{u'_1,\ldots,u'_{h'}\}$.
Let $G$ be an input to $\wHom{H}$ with connected components $G^1,\ldots,G^\kappa$.
Clearly, $Z_H(G)=\prod_{i=1}^\kappa Z_{H}(G^i)$.
Also, if $G^i$ is non-bipartite then $Z_{H}(G^i)=0$.
Suppose that $G^i$ is a connected bipartite graph with 
bipartition $(V,V')$ where $V=\{v_1,\ldots,v_n\}$
and $V'=\{v'_1,\ldots,v'_{n'}\}$.
Then 
$$Z_{H}(G^i) = 
\prod_{j=1}^n\sum_{c=1}^{h} w_{v_j}(u_c)
\prod_{j'=1}^{n'}\sum_{c'=1}^{h'} w_{v'_{j'}}(u'_{c'})
+
\prod_{j=1}^{n'}\sum_{c=1}^{h} w_{v'_{j}}(u_c)
\prod_{j'=1}^{n}\sum_{c'=1}^{h'} w_{v_{j'}}(u'_{c'})
.$$
 
In the context of this paper, 
where $H$ is a tree, we can draw  the following concluson.
\begin{observation}
\label{obs:star}\label{obs:triv}
Suppose that $H$ is a star. Then
 $\wHom{H}$ is in $\FP$.
\end{observation} 
 
\subsection{Trees with intermediate complexity}

The purpose of this section is to prove Lemma~\ref{lem:intermediate}, which says
that if $H$ is a tree that is not a star and has no induced~$J_3$
then $\BIS \APeq\nHom{H}$ and  $\BIS\APeq \wHom{H}$.
The main work of the section is in the proof of Lemma~\ref{lem:intermediate}, but
first we need some existing results. In particular, Lemma~\ref{lem:kelk} below is due to Kelk, and
Lemma~\ref{lem:CSP}   is an easy consequence of earlier work
by the authors and their coauthors on counting CSPs.
We have chosen to include a proof sketch of the former because the work
of Kelk is   unpublished~\cite{Kelk} and a proof of the latter because 
we did not state or prove it explicitly in earlier work, and it 
might
be rather difficult for the reader to
see why it is implied by  that work.

If $H$ is a tree
with no induced~$P_4$ then it is a 
star, so, by
Observation~\ref{obs:triv},
$\wHom{H}$ is in $\FP$.
On the other hand,  the following lemma shows that if $H$ contains an induced~$P_4$
then even the unweighted problem $\nHom{H}$ is $\BIS$-hard. 
To motivate  the lemma, suppose
that $H$ contains an induced~$P_4$. Then it
is a bipartite graph which is not complete, so
by Goldberg at al.~\cite[Theorem 2]{GKP} the (uniform) sampling 
problem for $H$-colourings of a graph 
is as hard as the sampling problem for independent sets in a bipartite graph.
This is not quite the result we are seeking, but it is close in spirit,
given the close connection between sampling and approximate counting.  
The following lemma, which is a special case of \cite[Lemma 2.19]{Kelk},
is exactly what we need. 

\begin{lemma} [Kelk] 
\label{lem:kelk}
Let $H$ be a tree containing an induced~$P_4$.
Then $$\BIS \APred \nHom{H}.$$
\end{lemma}
 
 \begin{proof} (Proof sketch) 
 We will not give a complete proof of Lemma~\ref{lem:kelk} since it is a special case of
a lemma of Kelk, but here is a sketch to give the reader a high-level idea of
the construction.
Let $\Delta$ be the maximum degree of vertices of~$H$ and let $\Delta'\leq \Delta$
be the maximum degree taken by a neighbour of a degree-$\Delta$ vertex in~$H$.
Note that $\Delta'\geq2$ since $H$ cannot be a star.
Let $(c,c')$ be any edge in~$H$ with $d_H(c)=\Delta$ and $d_H(c')=\Delta'$.
Let $N_c$ be the set $\Gamma_H(c)-\{c'\}$ and
let $N_{c'} = \Gamma_H(c')-\{c\}$. Since $H$ is a tree, there are no
edges in $H$ between $N_c$ and $N_{c'}$.
Now consider a connected instance $G$ of $\BIS$ with 
bipartition $V(G)=(V,V')$. 
Let $G'$ be the bipartite
graph with vertex set
$V(G)\cup \{C,C'\}$ (where~$C$ and~$C'$ are new vertices that are not in $V(G)$)
and edge set $E(G) \cup \{(C,C')\} \cup \{C\}\times V' \cup \{C'\} \times V$.
Consider an $H$-colouring $\sigma$ of~$G$ 
with  
$\sigma(C)=c$ and $\sigma(C')=c'$.
(Standard constructions can be used to augment $G'$ so that 
almost all homomorphisms to $H$ have this property.)
For every vertex $v\in V$, $\sigma(v) \in N_{c'} \cup \{c\}$
and for every vertex $v'\in V'$, $\sigma(v') \in N_c \cup \{c\}$.
Also, 
$\{v \in V \mid \sigma(v)\in N_{c'} \} \cup \{ v'\in V' \mid \sigma(v')\in N_c\}$
is an independent set of~$G$. 
Thus, there is an injection from independent sets of~$G$
into these $H$-colourings of~$G'$.
Standard tricks can be used to adjust
the construction so that almost all of the homomorphisms correspond to
\emph{maximum} independent sets of~$G$ and so that all maximum independent
sets correspond to approximately the same number of homomorphisms. The proof follows
from the fact that counting maximum independent sets in a bipartite graph 
is equivalent to~$\BIS$~\cite{APred}.
\end{proof}

As mentioned above, the main result of this section is Lemma~\ref{lem:intermediate}, which will be presented below.
Its proof relies on earlier work on
counting \emph{constraint satisfaction problems} (CSPs).
Suppose that $x$ and $x'$ are Boolean variables.
An assignment $\sigma: \{x,x'\}\to \{0,1\}$ is
said to satisfy the implication constraint 
$\IMP(x,x')$ if
$(\sigma(x),\sigma(x'))$ is in $\{ (0,0),(0,1),(1,1)\}$.
The idea is that ``$\sigma(x)=1$'' implies ``$\sigma(x')=1$''.
The assignment~$\sigma$ is said to satisfy the ``pinning'' constraint $\delta_0(x)$
if $\sigma(x)=0$ and the pinning constraint $\delta_1(x)$ if $\sigma(x)=1$.
If $X$ is a set of Boolean variables then
a set~$C$ of $\{\IMP,\delta_0,\delta_1\}$ constraints 
on~$X$ is a set of constraints of the form
$\delta_0(x)$, $\delta_1(x)$ and $\IMP(x,x')$ for $x$ and $x'$ in $X$. 
The set $S(X,C)$ of \emph{satisfying assignments} 
is the set of all assignments $\sigma: X \to \{0,1\}$ which
simultaneously satisfy all of the constraints in~$C$.
We will consider the following computational problem.
\begin{description}
\item[Problem] $\nCSP(\IMP,\delta_0,\delta_1)$.
\item[Instance] A set $X$ of Boolean variables and a set $C$ of 
$\{\IMP,\delta_0,\delta_1\}$ constraints on~$X$.
 \item[Output] $|S(X,C)|$.
\end{description}

We will also consider the following weighted version of $\nCSP(\IMP)$.
Suppose, for each $x\in X$, that $\gamma_x:\{0,1\} \rightarrow \posQ$ 
is a weight function. For 
an indexed set
$\gamma(X) = \{\gamma_x \mid x\in X\}$
of weight functions, let
$$Z(X,C,\gamma) = \sum_{\sigma\in S(X,C)} \prod_{x\in X} \gamma_x(\sigma(x)).$$
\begin{description}
\item[Problem] $\nCSP^*(\IMP,\delta_0,\delta_1)$.
\item[Instance] A set $X$ of Boolean variables, a set $C$ of  $\{\IMP,\delta_0,\delta_1\}$ constraints 
on  $X$, and 
an indexed set
$\gamma(X)$ of weight functions.
\item[Output]  $Z(X,C,\gamma)$.
\end{description}

 We will use the following lemma, which  follows from earlier work
 on counting CSPs.

 \begin{lemma}
 \label{lem:CSP}
 $\nCSP^*(\IMP,\delta_0,\delta_1) \APeq \BIS$.
 \end{lemma}
 
 \begin{proof}
Dyer, Goldberg, and Jerrum~\cite[Theorem 3]{trichotomy}
shows that $\nCSP(\IMP,\delta_0,
\delta_1) 
\APeq \BIS$. 
$\nCSP(\IMP,\delta_0,\delta_1)$ trivially reduces to 
$\nCSP^*(\IMP,\delta_0,\delta_1)$ since it is a special case.
Thus, it suffices to give an AP-reduction from
$\nCSP^*(\IMP,\delta_0,\delta_1)$ to $\nCSP(\IMP,\delta_0,\delta_1)$.
The idea behind the construction that we use comes from
Bulatov et al.~\cite[Lemma 36, Item~(i)]{LSM}.
We give the details in order to translate the construction into the current context.
  
Let $(X,C,\gamma)$ be an instance of $\nCSP^*(\IMP,\delta_0,\delta_1)$.
We can assume without loss of generality that 
all of the weights $\gamma_x(b)$ are 
positive integers
by multiplying 
all of the weights by the product of the denominators.
The construction that follows is not difficult but
the details are a little bit complicated, so
we use the following running example to illustrate.
Let $X=\{y,z\}$, $C = \IMP(y,z)$, 
$\gamma_y(0)=5$, $\gamma_y(1) = 2$,
$\gamma_z(0)=1$ and $\gamma_z(1)=1$.

For every variable $x\in X$, consider the weight function~$\gamma_x$. 
Let 
$k_x =  \max(\lceil \lg \gamma_x(0) \rceil ,\lceil \lg \gamma_x(1) \rceil)$.
For every $b\in \{0,1\}$,
write the bit-expansion of 
$\gamma_x(1\oplus b)$ 
as
 $$
\gamma_x(1\oplus b)
= a_{x,b,0} + a_{x,b,1} 2^1 + \cdots + a_{x,b,k_{x}} 2^{k_{x}},$$
where each $a_{x,b,i}\in \{0,1\}$. 
Note that $\gamma_x(1\oplus b)>0$ 
so there is at least one~$i$ with $a_{x,b,i}=1$.
Let 
$\min_{x,b} = \min\{i \mid a_{x,b,i}=1\}$
and $\max_{x,b} = \max \{ i \mid a_{x,b,i} = 1\}$.
If $i<\max_{x,b}$ and $a_{x,b,i}=1$ then let $\text{next}_{x,b,i} = \min\{j>i \mid a_{x,b,j}=1\}$.
If $i>\min_{x,b}$ and $a_{x,b,i}=1$ then let $\text{prev}_{x,b,i} = \max\{j<i \mid a_{x,b,j}=1\}$.
For the running example,
\begin{itemize}
\item 
$k_y= \lceil \lg 5 \rceil = 3$ and $k_z = \lceil \lg 1 \rceil =0$.
\item 
For the variable $y$, taking $b=0$ we have
$\gamma_y(1\oplus 0) = 2^1$ so
$a_{y,0,0}=0$, $a_{y,0,1}=1$, and $a_{y,0,2}= a_{y,0,3}=0$. Also,
$\min_{y,0}=1=\max_{y,0}$.
\item 
Similarly,  taking $b=1$ gives
$\gamma_y(1\oplus 1) = 2^0+2^2$ so
$a_{y,1,0}=1$, $a_{y,1,1}=0$, $a_{y,1,2}=1$ and $a_{y,1,3}=0$. 
Thus $\min_{y,1}=0$ and $\max_{y,1}=2$.
Then $\text{next}_{y,1,0}=2$
and $\text{prev}_{y,1,2}=0$.
\item
Finally, for the variable $z$  and $b\in \{0,1\}$,  we have
$\gamma_z(1\oplus b)=2^0$
so 
$a_{z,b,0}=1$ and $\min_{z,b}=0=\max_{z,b}$.
\end{itemize}

Now for every $x\in X$, for every $ i \in \{1,\ldots,k_x\}$
and every $b\in\{0,1\}$ with $a_{x,b,i}=1$
let $A_{x,b,i}$ be the set of $i+2$
variables $\{x_{b,i,1},\ldots,x_{b,i,i}\} \cup \{ L_{x,b,i},R_{x,b,i}\}$.
Let $C_{x,b,i}$ be the set of implication constraints
$\bigcup_{j\in[i]} \{\IMP(L_{x,b,i},x_{b,i,j}),\IMP(x_{b,i,j},R_{x,b,i})\}$.
Note that there are $2^i+2$ satisfying assignments to the $\nCSP$ instance $(A_{x,b,i},C_{x,b,i})$:
one with $\sigma(L_{x,b,i})=\sigma(R_{x,b,i})=0$,
one with $\sigma(L_{x,b,i})=\sigma(R_{x,b,i})=1$, and
$2^i$ with $\sigma(L_{x,b,i})=0$ and $\sigma(R_{x,b,i})=1$.
The point here is that the sets $A_{x,b,i}$ 
will be combined for different values of~$i$.
The satisfying assignments
with  $\sigma(L_{x,b,i})=\sigma(R_{x,b,i})=0$
will correspond to contributions from a different index $i'>i$
and the satisfying assignments
with $\sigma(L_{x,b,i})=\sigma(R_{x,b,i})=1$
will correspond to contributions from a different index $i'<i$.
There are exactly $2^i$ satisfying assignments with 
$\sigma(L_{x,b,i})=0$ and $\sigma(R_{x,b,i})=1$
and these will correspond to the 
$a_{x,b,i} 2^i$ summand in the bit-expansion of $\gamma_x(1\oplus b)$.
For the running example,
\begin{itemize}
\item
for the variable $y$ and for $b=0$ and $i=1$ 
we have $A_{y,0,1} = \{y_{0,1,1} \} \cup \{ L_{y,0,1},R_{y,0,1}\}$.
Then $C_{y,0,1}$ contains
$  \{\IMP(L_{y,0,1},y_{0,1,1}),\IMP(y_{0,1,1},R_{y,0,1})\}$ and there are  $2+2^1=4$ solutions.
\item
For the variable $y$ and for $b=1$ and
$i=2$
we have  $A_{y,1,2} = \{y_{1,2,1}, y_{1,2,2}\} \cup \{ L_{y,1,2},R_{y,1,2}\}$.
 Then $C_{y,1,2}$ contains the constraints
   $\IMP(L_{y,1,2},y_{1,2,1})$, $\IMP(y_{1,2,1},R_{y,1,2})$,
 $\IMP(L_{y,1,2},y_{1,2,2})$, and $\IMP(y_{1,2,2},R_{y,1,2})$
 and there are $2+2^2=6$~solutions. 
 \end{itemize}

We now add some constraints corresponding to the $i=0$ case above.
For every $x\in X$
and every $b\in \{0,1\}$ with $a_{x,b,0}=1$
let $A_{x,b,0}$ be the set of
variables $\{ L_{x,b,0},R_{x,b,0}\}$.
Let $C_{x,b,0}$ be the set containing the constraint
$\IMP(L_{x,b,0},R_{x,b,0})$.
Note that there are $2^0+2=3$ satisfying assignments to the $\nCSP$ instance $(A_{x,b,0},C_{x,b,0})$:
one with $\sigma(L_{x,b,0})=\sigma(R_{x,b,0})=0$,
one with $\sigma(L_{x,b,0})=\sigma(R_{x,b,0})=1$, and
$2^0=1$ with $\sigma(L_{x,b,0})=0$ and $\sigma(R_{x,b,0})=1$.
For the running example,
\begin{itemize}
\item $A_{y,1,0} = \{ L_{y,1,0},R_{y,1,0}\}$
and $C_{y,1,0} = \{ \IMP(L_{y,1,0},R_{y,1,0})\}$.
\item  For $b\in \{0,1\}$, $A_{z,b,0} = \{ L_{z,b,0},R_{z,b,0}\}$
and $C_{z,b,0} = \{ \IMP(L_{z,b,0},R_{z,b,0})\}$.
\end{itemize}

 Now 
for every $x\in X$ and $b\in\{0,1\}$
let 
$C'_{x,b}$  
be the set of constraints
forcing equality of
$\sigma(R_{x,b,i})$ and $\sigma(L_{x,b,j})$ when $i$ and $j$ are adjacent one-bits in the bit-expansion
of 
$\gamma_x(1\oplus b)$.
In particular, 
$$C'_{x,b} = \bigcup_{\text{next}_{x,b,i}=j, \text{prev}_{x,b,j}=i}
\{ \IMP(R_{x,b,i},L_{x,b,j}), \IMP(L_{x,b,j},R_{x,b,i}) \}
$$
For the running example,
\begin{itemize}
\item
$C'_{y,0} =  C'_{z,0} =  C'_{z,1} =  \emptyset$  since
these variables have  only one positive coefficient in the bit expansion. 
\item For the variable $y$ and $b=1$ the relevant non-zero coefficients
are  $i=0$ and $j=2$ 
 so we get
$$C'_{y,1} =  
\{  
\IMP(R_{y,1,0},L_{y,1,2}), \IMP(L_{y,1,2},R_{y,1,0}) \}.
$$ 
\end{itemize}

Now consider $x\in X$. Let
$C''_{x,0}=C'_{x,0}  \cup \{\delta_0(L_{x,0,\min_{x,0}})\}$ and
  let  
$C''_{x,1} = C'_{x,1} \cup \{\delta_1(R_{x,1,\max_{x,1}})\}$. 
 For $x\in X$ and $b\in \{0,1\}$ let
$$A_{x,b} = \bigcup_{i\in
\{0,\ldots,k_x\},
a_{x,b,i}=1} A_{x,b,i}$$
and let 
$$C_{x,b} = 
C''_{x,b} 
\cup 
\bigcup_{i\in\{0,\ldots,k_x\},
a_{x,b,i}=1} C_{x,b,i}.
$$ 
Now will show that 
there are  
$\gamma_x(1)$
satisfying assignments to the $\nCSP$ instance $(A_{x,0},C_{x,0})$  which have the property that
$\sigma(R_{x,0,\max_{x,0}})=1$
and one satisfying assignment in which 
$\sigma(R_{x,0,\max_{x,0}})=0.$
To see this, note that the constraint
$\delta_0(L_{x,0,\min_{x,0}})$
forces 
$\sigma(L_{x,0,\min_{x,0}})=0$.
If $\sigma(R_{x,0,\max_{x,0}})=0$ then
all of the variables in $A_{x,0}$ are assigned spin~$0$ by~$\sigma$.
Otherwise, there is exactly one~$i$ with
$a_{x,0,i}=1$ and
$\sigma(L_{x,0,i})=0$ and $\sigma(R_{x,0,i})=1$.
As we noted above, there are $2^i$ assignments to the variables in~$A_{x,b,i}$.
But $\sum_{i: a_{x,0,1}=i} 2^i = \gamma_x(1)$, as required.
Similarly,  
there are  
$\gamma_x(0)$
satisfying assignments to the $\nCSP$ instance $(A_{x,1},C_{x,1})$   in which
$\sigma(L_{x,1,\min_{x,1}})=0$ and
there is one satisfying assignment in which 
$\sigma(L_{x,1,\min_{x,1}})=1.$
Let us quickly apply this to the running example.
\begin{itemize}
\item  Taking variable $y$ and $b=0$ we have
$A_{y,0} = A_{y,0,1}$  and
$C''_{y,0} = \{\delta_0(L_{y,0,1})\} \cup C_{y,0,1}$. Then
$\max_{y,0}=1$.
From above, there is one solution~$\sigma$ with $\sigma(R_{y,0,\max_{y,0}})=0$
and there are $2^1=\gamma_y(1)$ solutions 
$\sigma$
with $\sigma(R_{y,0,\max_{y,0}})=1$.
\item Taking variable $y$ and $b=1$ we have 
$$A_{y,1} = A_{y,1,0} \cup A_{y,1,2}$$
and
$$C''_{y,1} = \{ \delta_1(R_{y,1,2}),
\IMP(R_{y,1,0},L_{y,1,2}), \IMP(L_{y,1,2},R_{y,1,0})
\} \cup C_{y,1,0} \cup C_{y,1,2} .$$
There is one solution $\sigma$ with 
$\sigma(L_{y,1,0})=1$.
There are $2^0+2^2=\gamma_y(0)$ solutions $\sigma$ with 
$\sigma(L_{y,1,0})=0$. 
\item Taking variable $z$   we have
$A_{z,b} = A_{z,b,0} = \{L_{z,b,0},R_{z,b,0}\}$. Then,
taking $b=0$,
$C_{z,0} = \{ \delta_0(L_{z,0,0}),\IMP(L_{z,0,0},R_{z,0,0})\}$.
so there is 
$2^0=1=\gamma_z(1)$ assignment with $\sigma(R_{z,0,0})=1$
and one with $\sigma(R_{z,0,0})=0$.
Taking $b=1$,
 $C_{z,1} = \{\delta_1(R_{z,1,0}),\IMP(L_{z,1,0},R_{z,1,0})\}$
so there is
$2^0=1=\gamma_z(0)$ assignment with $\sigma(L_{z,1,0})=0$
and one with $\sigma(L_{z,1,0})=1$.

\end{itemize}

Finally, consider $x\in X$. Let $C_x$ be
the set of constraints containing the four 
implications $\IMP(x,R_{x,0,\max_{x,0}})$,
$\IMP(R_{x,0,\max_{x,0}},x)$,
$\IMP(x,L_{x,1,\min_{x,1}})$,
and $\IMP(L_{x,1,\min_{x,1}},x)$.
Now there are $\gamma_x(1)$
solutions to $(A_{x,0} \cup A_{x,1} \cup \{x\},C_{x,0} \cup C_{x,1} \cup  C_x)$
with $\sigma(x)=1$ and $\gamma_x(0)$ solutions with $\sigma(x)=0$.
Thus, we have simulated
the weight function $w_x$ with $\{\IMP,\delta_0,\delta_1\}$ constraints.
For the running example, 
\begin{itemize}
\item first consider the variable $y$.
\begin{itemize}
\item 
With $\sigma(y)=1$ 
the constraints in $C_y$
force
$\sigma(R_{y,0,\max_{y,0}})=1$ 
which, from above, gives $\gamma_y(1)$ solutions to $(A_{y,0},C_{y,0})$.
The constraints in 
$C_y$ also force 
$\sigma(L_{y,1,\min(y,1)})=1$,
which, from above, gives one solution to  $(A_{y,1},C_{y,1})$.
\item
With $\sigma(y)=0$ the 
constraints in $C_y$ force
$\sigma(R_{y,0,\max_{y,0}})=0$
so there is only one solution to $(A_{y,0},C_{y,0})$.
The constraints in $C_y$ also force
$\sigma(L_{y,1,\min(y,1)})=0$ so there are $\gamma_y(0)$ solutions to $(A_{y,1},C_{y,1})$.
\end{itemize}
\item  The argument for variable~$z$ is similar.
\end{itemize}

Thus, the correct output for the $\nCSP^*(\IMP,\delta_0,\delta_1)$
instance $(X,C,\gamma)$ is
same as the correct output for the $\nCSP(\IMP,\delta_0,\delta_1)$ instance
obtained from $(X,C,\gamma)$ by 
adding new variables and constraints to simulate each weight function $\gamma_x$. 
\end{proof}

We can now prove the main lemma of this section.

\begin{lemma}
\label{lem:intermediate}
Suppose that $H$ is a tree
which is not a star and which
has no induced~$J_3$.
Then $$\BIS \APeq\nHom{H} \mbox{ and } \BIS\APeq \wHom{H}.$$
\end{lemma}

\begin{proof}
$\nHom{H}$ is a special case of $\wHom{H}$
so it is certainly
AP-reducible to
$\wHom{H}$. By Lemma~\ref{lem:kelk},
$\BIS$ is AP-reducible to $\nHom{H}$ and therefore it is
AP-reducible to $\wHom{H}$.
So it
suffices to give an AP-reduction from $\wHom{H}$ to $\BIS$.
Applying Lemma~\ref{lem:CSP}, it
suffices to give an AP-reduction from $\wHom{H}$ to $\nCSP^*(\IMP,\delta_0,\delta_1)$.

In order to do the reduction, we will order the vertices of~$H$
using the fact that it has no induced~$J_3$.
(This ordering is similar the one arising from the ``crossing property'' of the authors 
that is mentioned in \cite[Section 7.3.3]{Kelk}.)   
A  ``convex ordering'' of a connected bipartite graph with bipartition $(U,U')$
with $|U|=h$ and $|U'|=h'$ and edge set $E\subseteq U\times U'$
is a pair of bijections $\pi:U \rightarrow [h]$ and
$\pi':U' \rightarrow [h']$ such that
there are 
monotonically non-decreasing functions
functions $m:[h]\to[h']$, $M:[h]\to[h']$, $m':[h']\to[h]$ and $M':[h']\to[h]$ satisfying the
following conditions.
\begin{itemize} 
\item If $\pi(u)=i$
then $\{\pi'(u') \mid (u,u')\in E \} = \{ \ell \in [h'] \mid m(i) \leq \ell \leq M(i)\}$.
\item If $\pi'(u')=i$
then $\{\pi(u) \mid (u,u')\in E \} = \{ \ell \in [h] \mid m'(i) \leq \ell \leq M'(i)\}$.
\end{itemize}
 
The purpose of~$\pi$ and~$\pi'$ is just to put the
vertices in the correct order.
For example, in Figure~\ref{fig:referee},
\begin{figure}
\centering{
\begin{tikzpicture}[fill=white,scale=0.4,
line width=0.5pt,inner sep=1pt,minimum size=2.5mm]
\pgfsetxvec{\pgfpoint{1.7cm}{0cm}}
\pgfsetyvec{\pgfpoint{0cm}{1.7cm}}
\path 
(0,0) node [draw,fill,circle,minimum size=0.6cm](u4){ $4$}
(0,2) node [draw,fill,circle,minimum size=0.6cm](u3){ $3$}
(0,4) node [draw,fill,circle,minimum size=0.6cm](u2){ $2$}
(0,6) node [draw,fill,circle,minimum size=0.6cm](u1){ $1$}
(2,0) node [draw,fill,circle,minimum size=0.6cm](up3){ $3$}
(2,2) node [draw,fill,circle,minimum size=0.6cm](up2){ $2$}
(2,4) node [draw,fill,circle,minimum size=0.6cm](up1){ $1$}
;
\draw [-] (u1)--(up1)--(u2)--(up1)--(u3)--(up2)--(u3)--(up3)--(u4);
\path (0,7) node (){$U$};
\path (2,5) node () {$U'$};
\end{tikzpicture}
}
\caption{An example of a convex ordering}
\label{fig:referee}
\end{figure}
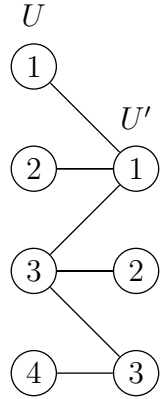
$\pi$ is the identity map on the set $U=\{1,2,3,4\}$ and
$\pi'$ is the identity map on the set $U'=\{1,2,3\}$.
Vertex~$3$ in~$U$ is connected to the sequence containing
vertices $1$, $2$ and $3$ in~$U'$,
so $m(3)=1$ and $M(3)=3$.
Every other vertex in~$U$ has degree~$1$ and in particular
$m(1)=M(1)=1$,
$m(2)=M(2)=1$ and
$m(4)=M(4)=3$.
Similarly, vertex~$1$ in~$U'$ is attached to
the sequence containing vertices $1$, $2$ and $3$ in $U$
so $m'(1)=1$ and $M'(1)=3$ but
$m'(2)=M'(2)=3$ and $m'(3)=M'(3)=4$.

To see that a convex ordering of~$H$ 
always
exists, consider the following algorithm.
The input is a tree~$H$ with 
no induced~$J_3$, a bipartition
$(U,U')$ of the vertices of~$H$, and a distinguished leaf~$u\in U$
whose parent~$u'$ is adjacent to at most one non-leaf. 
(Note that such a leaf~$u$ always exists since $H$ is a tree.)
The output is a convex ordering of~$H$ in which $\pi(u)=h$ and  $\pi'(u')=h'$.
Here is what the algorithm does.
If all of the neighbours of~$u'$ are leaves, then $h'=1$ so
take 
any bijection $\pi$ from $U-\{u\}$ to $[h-1]$ and set $\pi(u)=h$ and $\pi'(u')=h'$. Return this output.
Otherwise, let $u''$ be the neighbour of $u'$ that is not a leaf.
Let $H'$ be the graph
formed from $H$ by removing all of the $d_H(u')-1$ neighbours of $u'$ other than $u''$.
Since $H$ has no induced $J_3$, the graph $H'$ has the following property:
$u'$ is a leaf whose parent, $u''$, is adjacent to at most one non-leaf.
Recursively, construct a convex ordering for $H'$ in which $\pi(u')=h'$ 
and $\pi(u'')=h-(d_H(u')-1)$.
Extend $\pi$ by assigning values to the leaf-neighbours of~$u'$, ensuring that $\pi(u)=h$.

We will now show how to reduce $\wHom{H}$ to $\nCSP^*(\IMP,\delta_0,\delta_1)$.
Let $G$ be a connected bipartite graph
with bipartition $(V,V')$ and let $W(G,H)$ be 
an indexed set
of weight functions.
Let 
$$Z'_{H}(G,W(G,H))  = \sum_{\sigma\in \Hom GH\text{ with $\sigma(V)\subseteq U$}}\,
\prod_{v\in V(G)} w_v(\sigma(v))$$
and let
$$Z''_{H}(G,W(G,H))  = \sum_{\sigma\in \Hom GH\text{ with $\sigma(V)\subseteq U'$}}\,
\prod_{v\in V(G)} w_v(\sigma(v)).$$
 Clearly, $Z_{H}(G,W(G,H)) = Z'_{H}(G,W(G,H))+Z''_{H}(G,W(G,H))$.
 We will show how to reduce the computation of $Z'_{H}(G,W(G,H))$,
 given the input $(G,W(G,H))$,
 to the problem $\nCSP^*(\IMP,\delta_0,\delta_1)$.
 In the same way, we can reduce the computation of $Z''_{H}(G,W(G,H))$ to $\nCSP^*(\IMP,\delta_0,\delta_1)$.
 
Since we are considering assignments which map 
$V$ to $U$ and $V'$ to $U'$, the
vertices in~$U$ will not get mixed up with the vertices in~$U'$.
We can simplify the notation by relabelling the vertices so that $\pi$ and~$\pi'$
are the identity permutations.  Then,
given the convex ordering property, we can assume that 
$U=[h]$ and that $U'=[h']$
and that we have monotonically non-decreasing functions
functions $m:[h]\to[h']$, $M:[h]\to[h']$, $m':[h']\to[h]$ and $M':[h']\to[h]$ such that
\begin{itemize} 
\item  for $i\in U$, $\Gamma_H(i) = \{ \ell \in [h'] \mid m(i) \leq \ell \leq M(i)\}$, and
\item for $i\in U'$, $\Gamma_H(i) = \{ \ell \in [h] \mid m'(i) \leq \ell \leq M'(i)\}$.
\end{itemize}

A configuration $\sigma$ contributing to $Z'_{H}(G,W(G,H))$
is a map  from $V$ to $[h]$ together with a map from $V'$ to $[h']$
such that the following is true for every edge $(v,v')\in V\times V'$.
\begin{enumerate}[(1)]
\item \label{one} $m({\sigma(v)}) \leq \sigma(v') \leq M({\sigma(v)})$, and
\item  \label{two} $m'({\sigma(v')}) \leq \sigma(v) \leq M'({\sigma(v')})$. 
\end{enumerate} 

Since $m$, $M$, $m'$ and $M'$ are monotonically non-decreasing,
we can re-write the conditions  
in a less natural way which will be straightforward to apply below.
\begin{enumerate}[($1'$)]
\item \label{onep} $\sigma(v) \leq i$ implies $\sigma(v') \leq M(i)$,
\item \label{twop} $\sigma(v') \leq i'$ implies $\sigma(v) \leq M'(i')$,
\item \label{threep} $\sigma(v') \leq m(i)-1$ implies $\sigma(v) \leq i-1$, and
\item \label{fourp} $\sigma(v) \leq m'(i')-1$ implies $\sigma(v') \leq i'-1$. 
\end{enumerate}
Using monotonicity, (\ref{onep}$'$) and (\ref{twop}$'$) follow from the right-hand side of (\ref{one}) and (\ref{two}). 
Suppose that $\sigma(v') < m(i)$. Then the left-hand side of (\ref{one}) gives
$m(\sigma(v))< m(i)$, so by monotonicity, $\sigma(v)< i$. Equation~(\ref{threep}$'$) follows.
In the same way, Equation~(\ref{fourp}$'$) follows from the left-hand side of (\ref{two}).
Going the other direction, the right-hand sides of (\ref{one}) and (\ref{two}) follow
from (\ref{onep}$'$) and (\ref{twop}$'$).To derive the left-hand side of (\ref{one}), 
take the contrapositive of 
(\ref{threep}$'$), 
which says 
$\sigma(v) \geq i$ implies $\sigma(v') \geq m(i)$ then plug in $i=\sigma(v)$.
The derivation of the left-hand side of (\ref{two}) is similar.

We now construct an instance of $\nCSP^*(\IMP,\delta_0,\delta_1)$.
For each vertex $v\in V$ introduce Boolean variables $v_0,\ldots,v_{h}$.
Introduce constraints $\delta_0(v_0)$ and $\delta_1(v_{h})$ 
and, for every $i\in[h]$, $\IMP(v_{i-1},v_i)$.
For each vertex $v'\in V'$ introduces Boolean variables $v'_0,\ldots,v'_{h'}$.
Introduce constraints $\delta_0(v'_0)$ and $\delta_1(v'_{h'})$ 
and, for every $i'\in[h']$, $\IMP(v'_{i'-1},v'_{i'})$.

Now there is a one-to-one correspondence between
assignments $\sigma$
mapping $V$ to~$U$ and $V'$ to~$U'$,
and assignments $\tau$ to the Boolean variables that satisfy the above constraints.
In particular, 
$\sigma(v)=\min\{i \mid \tau(v_i)=1\}$.
Similarly, 
$\sigma(v')=\min\{i' \mid \tau(v'_i)=1\}$.

Now, $\sigma(v) \leq i$ is exactly equivalent to
$\tau(v_i) =1$. 
Thus, we can add the following further constraints to 
rule out assignments 
$\sigma$ that do not  
satisfy (\ref{onep}$'$), (\ref{twop}$'$), (\ref{threep}$'$) and (\ref{fourp}$'$).
Add all of the following constraints
where $v\in V$, $v'\in V'$, $i\in [h]$ and 
$i'\in [h']$:
$\IMP(v_{i},v'_{M(i)})$,
$\IMP(v'_{i'}, v_{M'(i')})$,
$\IMP(v'_{m(i)-1},v_{i-1})$, and
$\IMP(v_{m'(i')-1},v'_{i'-1})$. 
Now the assignments~$\tau$ of Boolean values to the variables satisfy all of the constraints if and only if
they 
correspond to assignments~$\sigma$ which satisfy 
 (\ref{onep}$'$), (\ref{twop}$'$), (\ref{threep}$'$) (\ref{fourp}$'$),
 and
  so should  contribute to
$$Z'_{H}(G,W(G,H))  = \sum_{\sigma\in \Hom GH\text{ with $\sigma(V)\subseteq U$}} \,
\prod_{v\in V(G)} w_v(\sigma(v)).$$

We will next construct weight functions for the 
instance of $\nCSP^*(\IMP,\delta_0,\delta_1)$ in order to reproduce the
effect of the weight functions in $W(G,H)$.

In order to avoid division by~$0$,
we first modify the construction.
Suppose that for some variable $v\in V$ and some $i\in [h]$,
$w_v(i)=0$.
Configurations $\sigma$ with $\sigma(v)=i$
make no contribution to  $Z'_{H}(G,W(G,H))$.
Thus, it does no harm to rule out such configurations by  
modifying the $\nCSP^*(\IMP,\delta_0,\delta_1)$ instance to ensure that
$\tau(v_i)=1$ implies $\tau(v_{i-1})=1$.
We do this by adding the constraint $\IMP(v_i,v_{i-1})$. Similarly, if
$w_{v'}(i')=0$ for $v'\in V$ and $i'\in[h']$ then we add the constraint $\IMP(v'_{i'},v'_{i'-1})$.

Once we've made this change, we can 
replace $W(G,H)$ with
an equivalent 
indexed
set of weight functions $W'(G,H)$ 
where $w'_v(i)=w_v(i)$ if $w_v(i)>0$ and $w'_v(i)=1$, otherwise. 

The weight functions 
for the $\nCSP^*(\IMP,\delta_0,\delta_1)$ instance
are then constructed as follows, for each $v\in V$.
For  
each $i\in[h]$, let 
$\gamma_{v_{i-1}}(0)=1$.
Let $\gamma_{v_h}(1)=w'_v(h)$.
For each $i\in [h-1]$, let 
$\gamma_{v_i}(1)=w'_v(i)/w'_v(i+1)$.
Note that $\gamma_{v_h}(0)$ and $\gamma_{v_0}(1)$ have not yet been
defined --- these values can be chosen arbitrarily. They will not be relevant
given the constraints $\delta_0(v_0)$ and $\delta_1(v_h)$.

Now if $\sigma(v)=i$ 
we have $\tau(v_0)=\cdots = \tau(v_{i-1})=0$ and
$\tau(v_i)=\cdots = \tau(v_h)=1$ 
so 
$\prod_j \gamma_{v_j}(\tau(v_j)) = w'_v(i)$, 
as required.
Similarly, for each $v'\in V'$, define the weight functions as follows.
For 
each $i\in[h']$, let 
$\gamma_{v'_{i-1}}(0)=1$.
Let $\gamma_{v'_{h'}}(1)=w'_{v'}(h')$.
For each $i\in [h'-1]$, let 
$\gamma_{v'_i}(1)=w'_{v'}(i)/w'_{v'}(i+1)$.
Using these weight functions, we obtain the desired reduction from  
the computation of $Z'_{H}(G,W(G,H))$ 
to $\nCSP^*(\IMP,\delta_0,\delta_1)$.
\end{proof}

\subsection{Intractable trees}

Lemma~\ref{lem:intermediate} shows that if $H$ has no induced~$J_3$
then $\wHom{H}$ is AP-reducible to $\BIS$.
The purpose of this section is to 
prove Lemma~\ref{lem:hardweighted}, below, which 
shows, by contrast, that
if $H$ does 
have an induced~$J_3$,
then $\wHom{H}$ is $\SAT$-hard.

In order to prepare for the proof of Lemma~\ref{lem:hardweighted}, we
introduce the  notion of a multiterminal cut.
Given a graph~$G=(V,E)$ 
with distinguished vertices~$\alpha$, $\beta$ and~$\gamma$, which we refer to as
``terminals'', a {\it multiterminal cut\/} is a set $E'\subseteq E$
whose removal disconnects the terminals in the sense that the graph
$(V,E\setminus E')$ does not contain a path between any two distinct
terminals. The size of the multiterminal cut is the number of edges
in~$E'$.  Consider the following computational problem.

\begin{description}
\item[Problem]  \MultiCutCount{$3$}. 
\item[Instance] A positive integer~$b$, a connected
graph $G=(V,E)$ and $3$ distinct vertices $\alpha$, $\beta$ and $\gamma$ from~$V$.
The input has the property that every multiterminal cut  has size at
least~$b$. 
\item[Output]
The number of size-$b$ multiterminal cuts for
$G$ with  terminals $\alpha$, $\beta$, and $\gamma$.
\end{description}

We will use the following technical lemma, which we used
before in~\cite{Ising} (without stating it formally).
\begin{lemma}
\label{lem:cut}
\MultiCutCount{$3$}  $\APeq \SAT$.
\end{lemma}

\begin{proof} 
This follows essentially from the  
proof of Dalhaus et al.~\cite{Dalhaus} that the decision version of 
 \MultiCutCount{$3$}
is NP-hard 
and from the fact \cite[Theorem 1]{APred} that
the NP-hardness of a decision problem implies that the corresponding counting 
problem is AP-interreducible with $\SAT$.
The details are given in \cite[Section 4]{Ising}.
\end{proof}

\begin{lemma}
\label{lem:hardweighted} Suppose that $H$ is a tree with an induced~$J_3$.
Then $$\SAT \APred\wHom{H}.$$
\end{lemma}

\begin{proof}
We will prove the lemma by giving an  AP-reduction from \MultiCutCount{$3$} to $\wHom{H}$.
The lemma will then follow from Lemma~\ref{lem:cut}.

Suppose that $H$ has an induced subgraph which is isomorphic to~$J_3$.
To simplify the notation, label the vertices and edges of~$H$ in such a way that
the induced subgraph is (identically) the graph $J$ depicted in Figure~\ref{fig:J}.
\begin{figure}
\centering{
 \begin{tikzpicture}[fill=white,scale=0.7,
line width=0.5pt,inner sep=1pt,minimum size=2.5mm]
\pgfsetxvec{\pgfpoint{1.7cm}{0cm}}
\pgfsetyvec{\pgfpoint{0cm}{1.7cm}}
\path 
(0,0) node [draw,fill,circle,minimum size=0.6cm](w){ $w$}
(1,0) node [draw,fill,circle,minimum size=0.6cm](y0){ $y_0$}
(-1,0) node [draw,fill,circle,minimum size=0.6cm](x0){ $x_0$}
(0,-1) node [draw,fill,circle,minimum size=0.6cm](z0){ $z_0$}
(-2,0) node [draw,fill,circle,minimum size=0.6cm](x1){ $x_1$}
(2,0) node [draw,fill,circle,minimum size=0.6cm](y1){ $y_1$}
(0,-2) node [draw,fill,circle,minimum size=0.6cm](z1){ $z_1$}
;
\draw [-] (x1)--(x0) -- (w) -- (y0) -- (y1);
\draw [-] (w)--(z0)--(z1); 

\end{tikzpicture}

 }
 \caption{The tree $J$.}
\label{fig:J}
\end{figure}
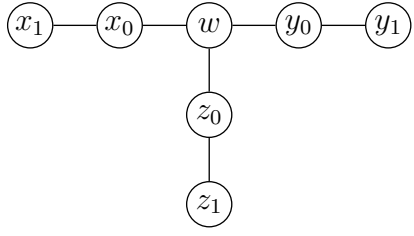

Let $b$, $G=(V,E)$, $\alpha$, $\beta$ and $\gamma$ be an input to \MultiCutCount{$3$}.
Let  
$s= 2 + |E(G)|+2|V(G)|$.
(The exact size of~$s$ is not important, but
it has to be at least this big to make the calculation work,
and it has to be at most a polynomial in the size of~$G$.)
Let $G'$ be the graph defined as follows.
First, let $V'(G)= \{(e,i) \mid e\in E,i\in[s]\}$.
Thus, $V'(G)$ contains $s$ vertices for each edge $e$ of~$G$.
Then let $G'$
be the graph with vertex set
$V(G') = V(G) \cup V'(G)$
and edge set 
 $$E(G') = \{(u,(e,i)) \mid u\in V(G), (e,i)\in V'(G),  
\mbox{and $u$ is  an endpoint of~$e$}  \}.$$

We will define weight functions $w_v$ for $v\in V(G')$
so that an approximation to
the number of size-$b$ multi-terminal cuts for $G$ with terminals $\alpha$, $\beta$ and $\gamma$
can be obtained from an approximation to
$Z_H(G',W(G',H))$. We start by defining the set of pairs $(v,c)\in V(G')\times V(H)$ 
for which we will specify $w_v(c)>0$. In particular, define the set $\Omega$ as follows.
$$\Omega = 
\{(\alpha,x_0),(\beta,y_0),(\gamma,z_0)\} 
\cup
\big((V(G)-\{\alpha,\beta,\gamma\}) \times \{x_0,y_0,z_0\} \big)
\cup
\left(V'(G) \times \{w,x_1,y_1,z_1\}\right).$$
Let $w_v(c)=1$ if $(v,c)\in \Omega$. Otherwise, let $w_v(c)=0$.

Thus,  $Z_H(G',W(G',H))$
is the number of homomorphisms $\sigma$ from~$G'$ to~$H$ with
$\sigma(V(G)) = \{x_0,y_0,z_0\}$,
$\sigma(V'(G)) \subseteq \{w,x_1,y_1,z_1\}$,
$\sigma(\alpha)=x_0$, $\sigma(\beta)=y_0$ and $\sigma(\gamma)=z_0$.
We will refer to these as ``valid'' homomorphisms. 
  
If $\sigma$ is  a  valid homomorphism, then let 
\begin{align*}
\bichrom{\sigma} =
\{ e \in E(G) \mid \quad & 
\mbox{the vertices of~$V(G)$ corresponding to } \\
& \mbox{the endpoints of~$e$ are mapped to different colours by~$\sigma$} 
\}.\end{align*}
Note that, for every valid homomorphism~$\sigma$,
$\bichrom{\sigma}$ is a multiterminal cut 
for the  graph~$G$ with terminals~$\alpha$, $\beta$ and~$\gamma$.

For every multiterminal cut $E'$, let $\components{E'}$
denote the number of components
in the graph $(V,E\setminus E')$.
For each multiterminal cut~$E'$,
let $Z_{E'}$ denote
the number of valid homomorphisms~$\sigma$ from~$G'$ to~$H$
such that $\bichrom{\sigma} = E'$.
From the definition of multiterminal cut, $\components{E'}\geq 3$. 
If $\components{E'}=3$ then
$$Z_{E'} =  2^{s(E(G)-E')}$$
since there are
 two choices for the colours of
each vertex $(e,i)$ with $e\in E(G)-E'$.
(Since the endpoints of each such edge~$e$
are assigned the same colour by~$\sigma$, the vertex $(e,i)$
can either be coloured~$w$, or it can be coloured with one other colour.)
Also, 
$$Z_{E'} \leq  2^{s(E(G)-E')} 3^{\components{E'}-3},$$ 
since  the component
of~$\alpha$ is mapped to~$x_0$ by~$\sigma$,
the component of~$\beta$ is mapped to~$y_0$,  
the component of~$\gamma$ is mapped to~$z_0$,  and each remaining
component  is mapped to  a colour in $\{x_0,y_0,z_0\}$.
 
Let $Z^*= 2^{s(E(G)-b)}  $.
If $E'$ has size~$b$ then 
$\components{E'}=3$. (Otherwise, there would be a smaller multiterminal cut,
contrary to the definition of \MultiCutCount{$3$}.)
So, in this case, 
\begin{equation}
Z_{E'} = Z^*.
\label{eq:smgoodcuts}
\end{equation}

If $E'$ has size $b'>b$ then 
$$Z_{E'} \leq 2^{s(E(G)-b')} 3^{\components{E'}-3}   
               = 2^{-s(b'-b)} 3^{\components{E'}-3} Z^*
               \leq 2^{-s} 3^{|V(G)|} Z^*.
$$
Clearly, there are at most $2^{|E(G)|}$   multiterminal cuts~$E'$.
So, using the definition of~$s$,  
\begin{equation}
\label{eq:smbigcuts}
\sum_{E' : |E'|>b} Z_{E'} \leq \frac{Z^*}{4}
\end{equation}  

From Equation~(\ref{eq:smgoodcuts}), we find that,
if  there are $N$ size-$b$ multiterminal cuts
then 
$$Z_H(G',W(G',H)) = N Z^* + \sum_{E' : |E'|>b} Z_{E'}  .$$
So applying 
Equation  (\ref{eq:smbigcuts})  , 
we get
$$ N \leq \frac{Z_H(G',W(G',H))}{Z^*} \leq N + \frac{1}{4}.$$

Thus, we have an AP-reduction
from \MultiCutCount{$3$}  to $\nHom{H}$.
To determine the accuracy with which~$Z(G)$ should be approximated
in order to achieve a given accuracy in the approximation to~$N$,
see the   proof of Theorem 3 of \cite{APred}.
\end{proof}

\section{Tree homomorphisms capture the ferromagnetic Potts model.}
\label{sec:potts} 
The problem $\nHom{H}$ counts colourings of a graph satisfying ``hard'' constraints:
two colours (corresponding to vertices of $H$) are either allowed on 
adjacent vertices of the instance or disallowed.  By contrast, the Potts model  
(to be described presently) is ``permissive'':  every pair of colours is allowed 
on adjacent vertices, but some pairs are favoured relative to others.  The strength 
of interactions between colours is controlled by a real parameter~$\gamma$.  
In this section, we will show that approximating the number of homomorphisms to $J_q$
is equivalent in difficulty to the problem of approximating the partition function of
the ferromagnetic $q$-state Potts model.  Since the latter problem is not 
known to be \BIS-easy for any $q>2$, we might speculate that approximating $\nHom{J_q}$ is
not \BIS-easy for any $q>2$.  If so, $J_3$ would be the smallest tree with this property.

It is interesting that, for fixed~$q$,
a continuously parameterised class of permissive problems can be shown to be
computationally equivalent to a single counting problem with hard constraints.  
Suppose, for example, that we wanted to investigate the possibility that computing the 
partition function of the $q$-state ferromagnetic Potts model formed a hierarchy 
of problems of increasing complexity with increasing~$q$.  
We could equivalently investigate the 
sequence of problems $\nHom{J_q}$, which seems intuitively to be an easier proposition.

We start with some definitions. 
Let $q$ be a positive integer. 
The $q$-state Potts model is 
a statistical mechanical model of Potts~\cite{Potts} which  
generalises the classical Ising model from 
two to $q$~spins.  
In this model, spins interact along edges of a graph~$G=(V,E)$.
The strength of each interaction is governed by
a parameter~$\gamma$ (a real number which is always at least~$-1$, and is greater than~$0$
in the   \emph{ferromagnetic} case 
which we study,
where like spins attract each other).
The $q$-state Potts partition function  is
defined as follows.
\begin{equation}\label{eq:PottsGph}
\ZPotts(G;q,\gamma) = 
\sum_{\sigma:V\rightarrow [q]}
\prod_{e=\{u,v\}\in E}
\big(1+\gamma\,\delta(\sigma(u) ,\sigma(v))\big),
\end{equation}
where  
$\delta(s,s')$ is~$1$ if $s=s'$, and is~$0$ otherwise.  

The Potts partition function is well-studied.
In addition to the complexity-theory literature mentioned below,
we refer the reader to Sokal's survey~\cite{Sokal05}.

In order to state our results in the strongest possible form,  
we use the notion of ``efficiently approximable real number'' from Section~\ref{sec:prelim}.
Recall that a real number $\gamma$ is efficiently approximable if there is an FPRAS for
the problem of computing it.
The notion of ``efficiently approximable'' is not important to the constructions below ---
the reader who prefers to 
assume that the parameters are rational
will still appreciate the essence of the reductions.

Let $q$ be a positive integer 
and let $\gamma$ be
a positive efficiently approximable real. Consider the
following computational problem, which is parameterised by~$q$ and~$\gamma$.
\begin{description}
\item[Problem] $\Potts(q,\gamma)$.
\item[Instance] Graph $\graph=(\graphvertices,\graphedges)$.
\item[Output]  $\ZPotts(\graph;q,\gamma)$.
\end{description} 
This problem may be defined more generally for 
non-integers~$q$ via the Tutte polynomial. 
We will use some results from
\cite{FerroPotts} which are more general, but we do not need
the generality here.

In an important paper, Jaeger, Vertigan and Welsh~\cite{JVW90} examined
the problem of evaluating the Tutte polynomial.
Their result gave a complete classification of the computational complexity of
$\Potts(q,\gamma)$.  
For every fixed positive integer~$q$, apart from the trivial $q=1$, and for
every fixed~$\gamma$,  
they
showed that this computational problem is \#P-hard.
When $q=1$ and $\gamma$ is rational, 
$\ZPotts(\graph;q,\gamma)$ can easily be exactly evaluated in polynomial time. 
The complexity of the approximation problem has also been
partially resolved. In the positive direction, Jerrum and Sinclair~\cite{JS93} gave an FPRAS
for the case $q=2$. In the negative direction, 
Goldberg and Jerrum~\cite{FerroPotts} showed that approximation is
$\BIS$-hard for every fixed $q>2$. They left open the question of whether approximating 
$\ZPotts(G;q,\gamma)$ is as easy as $\BIS$ (or whether it might be even harder).

In this paper, we show that  the approximation problem is equivalent in complexity 
to a tree homomorphism problem.
In particular, we show that
 $\Potts(q,\gamma)$ is AP-equivalent to the problem 
of approximately counting homomorphisms to the tree~$J_q$.

We first give an AP-reduction from
$\Potts(q,1)$  to $\nHom{J_q}$.

\begin{lemma}
\label{lem:tocol}
Let $q>2$ be a positive integer.
$$\Potts(q,1) \APred \nHom{J_q}.$$
\end{lemma}

\begin{proof}

Let $G$ be an   instance of $\Potts(q,1)$.
We can assume without loss of generality that $G$ is connected,
since it is clear from~(\ref{eq:PottsGph})
 that a graph~$G$ with connected components $G_1,\ldots,G_\kappa$ satisfies
$\ZPotts(G;q,\gamma)=\prod_{i=1}^{\kappa} \ZPotts(G_i;q,\gamma)$.

Let $G'$
be the 
graph with 
$$V(G') = V(G) \cup E(G)$$
and 
$$E(G') = \{(u,e) \mid u\in V(G), e \in E(G),  
\mbox{and $u$ is  an endpoint of~$e$}  \}.$$
$G'$ is sometimes referred to as the ``$2$-stretch'' of~$G$. 
For clarity, when we consider an element $e\in E(G)$ as a vertex 
of $G'$ (rather than an edge of $G$), 
we shall refer to it as the ``midpoint vertex corresponding to edge~$e$''.

 Let $s$ be an integer
 satisfying
 \begin{equation}
 \label{eq:firsts}
   8 q  {(q+1)}^{|V(G)|+|E(G)|} \leq {\left(\frac{q}{2}\right)}^s .
   \end{equation}
   For concreteness, take $s$ to be the smallest integer satisfying
   (\ref{eq:firsts}). 
The exact size of~$s$ is not so important. The calculation below 
relies on the fact that $s$ is large enough to satisfy~(\ref{eq:firsts}). On the other
hand, $s$ must be at most a polynomial in the size of~$G$, to make the reduction feasible.
  
We will construct an instance~$G''$ of $\nHom{J_q}$
by adding some gadgets to~$G'$.
Fix a vertex $v\in V(G)$.
Let $G''$ be the graph with 
$V(G'')=V(G) \cup E(G) \cup \{v_0,\ldots,v_s\}$
and $E(G'') = E(G') \cup \{(v,v_0)\} \cup 
\{(v_0,v_i) \mid i\in [s]\}$.
See Figure~\ref{fig:firstinstance}.
\begin{figure}
\centering{
 \begin{tikzpicture}[fill=white,scale=0.7,
line width=0.5pt,inner sep=1pt,minimum size=2.5mm]
\pgfsetxvec{\pgfpoint{1.7cm}{0cm}}
\pgfsetyvec{\pgfpoint{0cm}{1.7cm}}
 \path (0,-4.5) node   {$V(G)$};
\draw (0,-4.5) circle (2.5cm);
\draw (5,-4.5) circle (2.5cm);
\path (5,-4.5) node  {$E(G)$};
\draw [line width=3pt]  (1.5,-4.5) .. controls (2,-4)  and (3,-5) .. (3.5,-4.5);

\path  
(0,-3.7) node  [draw,fill,circle,minimum size=0.6cm](v){ $v$ } 
(5,0) node  [draw,fill,circle,minimum size=0.6cm](vp){ $v_0$ } 
[grow'=left, level distance=85mm,sibling distance=10mm]
   child {node [draw,fill,circle,minimum size=0.6cm] { $  v_{1}$} }  
   child {node [draw,fill,circle,minimum size=0.6cm] { $  v_{2}$} }  
   child {node  { } }  
   child {node  { } }  
   child {node [draw,fill,circle,minimum size=0.6cm] { $  v_{s}$} }  
;
\draw (v)--(vp);
\path (-0.1,0.0) node [draw,fill=black,circle,minimum size=0.2cm] {}; 
\path (-0.1,0.3) node [draw,fill=black,circle,minimum size=0.2cm] {}; 
\path (-0.1,0.6) node [draw,fill=black,circle,minimum size=0.2cm] {}; 
\end{tikzpicture}

 }
 \caption{The instance~$G''$.  The thick curved line between $V(G)$ and
 $E(G)$ indicates that  the edges in~$E(G')$ go between elements of~$V(G)$ and
 elements of~$E(G)$, but these are not shown.
 }
\label{fig:firstinstance}

\end{figure}
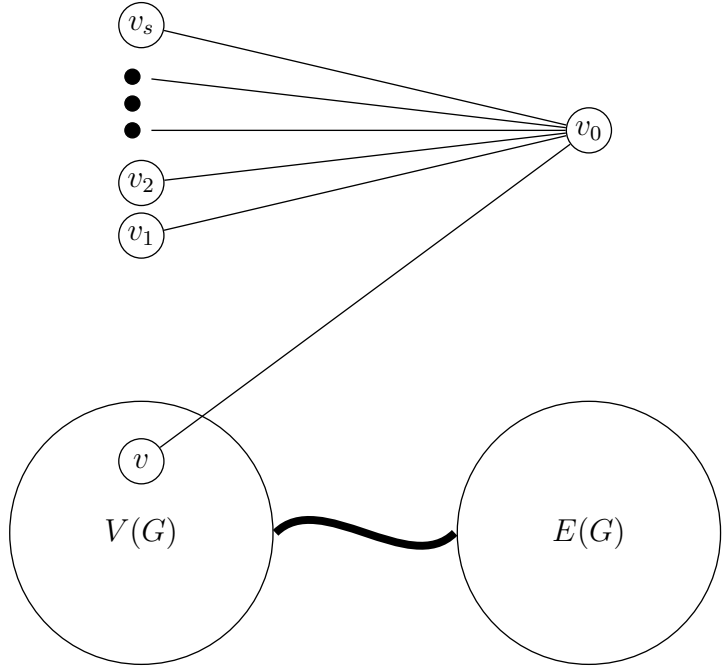

We say that a homomorphism~$\sigma$ from~$G''$ to~$J_q$ is \emph{typical} if
$\sigma(v_0)=w$.
Note that, in a typical homomorphism, 
every vertex in $V(G)$ is mapped by~$\sigma$ to
one of the colours from $\{c'_1,\ldots,c'_q\}$.
Let $Z_{J_q}^t(G'')$ denote the number of typical homomorphisms from~$G''$ to~$J_q$.

Given a mapping $\sigma: V(G) \rightarrow \{c'_1,\ldots,c'_q\}$,
the number of typical homomorphisms which induce this mapping is
$2^{\mathrm{mono}(\sigma)} q^s$, where $\mathrm{mono}(\sigma)$
is the number of edges $e\in E(G)$ whose endpoints in $V(G)$ are mapped to 
the same colour by~$\sigma$.
(To see this, note that there are two possible colours for 
 the midpoint vertices corresponding to such edges,
whereas the other 
midpoint vertices
have to be mapped to~$w$ by~$\sigma$.
Also, there are $q$ possible colours for each vertex in $\{v_1,\ldots,v_s\}$.)
Thus, using the definition~(\ref{eq:PottsGph}),
we conclude that
$$Z_{J_q}^t(G'') = \sum_{\sigma: V(G) \rightarrow \{c'_1,\ldots,c'_q\}} 2^{\mathrm{mono}(\sigma)} q^s
= q^s \ZPotts(G;q,1).$$

The number of atypical homomorphisms from~$G''$ to~$J_q$,
which we denote by $Z_{J_q}^a(G'')$, is
at most 
$2q 2^s {(q+1)}^{|V(G)|+|E(G)|}$. (To see this, note, that
there are $2q$ alternative colours for~$v_0$.
For each of these, there are at most~$2$ colours for each vertex in $\{v_1,\ldots,v_s\}$
and at most $q+1$ colours for each vertex in $V(G)\cup E(G)$.)
Using Equation~(\ref{eq:firsts}),
we conclude that   $Z_{J_q}^a(G'') \leq q^s/4$.
Since $Z_{J_q}(G'') = Z_{J_q}^t(G'') + Z_{J_q}^a(G'')$, we have

\begin{equation}
\label{done}
\ZPotts(G;q,1) \leq \frac{Z_{J_q}(G'')}{q^s} \leq \ZPotts(G;q,1) + \frac{1}{4}.
\end{equation}
  
Equation~(\ref{done}) guarantees that the construction is an AP-reduction from $\Potts(q,1)$ to
the problem
$\nHom{J_q}$. To determine the accuracy with which 
$Z_{J_q}(G'')$ should be approximated in order to achieve a given desired accuracy
in the approximation to $\ZPotts(G;q,1)$, see the proof of Theorem 3 of \cite{APred}.
\end{proof}
 
 In order to get a reduction going the other direction, we need to 
 generalise  the Potts partition function to a hypergraph version.
 Let $\hypergraph=(\hypervertices,\hyperedges)$ be a hypergraph
with vertex set $\hypervertices$ and hyperedge (multi)set~$\hyperedges$. 
 Let $q$ be a positive integer. The $q$-state Potts partition function of $\hypergraph$ is
defined as follows:
$$
\ZPotts(\hypergraph;q,\gamma) = 
\sum_{\sigma:\hypervertices\rightarrow [q]}
\prod_{\hyperedge\in\hyperedges}
\big(1+\gamma \delta(\{\sigma(\hypervertex) \mid \hypervertex\in \hyperedge\})\big),$$
where   $\delta(S)$ is~$1$ if its argument is a singleton and 0 otherwise.
Let $q$ be a positive integer and let $\gamma$ be a positive 
efficiently approximable real.
We consider the following computational problem, which is parameterised by~$q$ and~$\gamma$.
\begin{description}
\item[Problem] $\hPotts(q,\gamma)$.
\item[Instance]  A hypergraph $\hypergraph=(\hypervertices,\hyperedges)$.
\item[Output]  $\ZPotts(\hypergraph;q,\gamma)$.
\end{description} 
We start by reducing $\nHom{J_q}$ to
the problem of approximating the Potts partition function of a hypergraph
with parameters~$q$ and~$1$.

\begin{lemma}\label{lem:fromcol}
Let $q$ be a positive integer.
$$\nHom{J_q}  \APred \hPotts(q,1).$$ 
\end{lemma}

\begin{proof} 
We can assume without loss of generality
that the instance to $\nHom{J_q}$  is bipartite, since
otherwise the output is zero.
We can also assume that it is connected since 
a graph~$G$ with connected components $G_1,\ldots,G_\kappa$ satisfies
$Z_{J_q}(G) = \prod_{i=1}^\kappa Z_{J_q}(G_i)$.
Finally, it is easy to find a bipartition of a connected bipartite graph in polynomial
time, so we can assume without loss of generality that this is provided as
part of the input.

Let $B=(U,V,E)$ be a connected instance of $\nHom{J_q}$ consisting
of vertex sets~$U$ and~$V$ and edge set $E$ (a subset of $U\times V$).
Let $Z_{J_q}^U(B)$ be the number of homomorphisms
from $B$ to $J_q$ in which vertices in~$U$ are coloured with colours in~$\{c'_1,\ldots,c'_q\}$. 
Similarly, let $Z_{J_q}^V(B)$ be the number of 
homomorphisms
from $B$ to $J_q$ in which vertices in~$V$ are coloured with colours in~$\{c'_1,\ldots,c'_q\}$. 
Clearly, $Z_{J_q}(B) = Z_{J_q}^U(B) +   Z_{J_q}^V(B)$.
We will show how to approximate $Z_{J_q}^U(B)$ using an approximation oracle
for  
$\hPotts(q,1)$. The approximation of $Z_{J_q}^V(B)$ is similar.

The construction is straightforward.
For every $v\in V$, let $\Gamma(v)$ denote the set of neighbours of vertex~$v$ in~$B$.
Let $F = \{\Gamma(v), \mid v\in V\}$.
Let $H=(U,F)$ be an instance of $\hPotts(q,1)$.

The reduction is immediate, because
$Z_{J_q}^U(B) = \ZPotts(H;q,1)$.
To see this, note that
every configuration $\sigma: U \rightarrow \{c'_1,\ldots,c'_q\}$ 
contributes weight 
$2^{\mathrm{mono}(\sigma)}$
to  $\ZPotts(H;q,1)$, where
${\mathrm{mono}(\sigma)}$ is the number of hyperedges in~$F$ that are monochromatic in~$\sigma$.
Also, the configuration~$\sigma$
can be extended in 
exactly $2^{\mathrm{mono}(\sigma)}$ ways to homomorphisms 
from~$B$ to~$J_q$.
\end{proof}

The next step is to reduce the problem of approximating the Potts partition function of
a hypergraph to the problem of approximating the Potts partition function of
a \emph{uniform} hypergraph, which is a hypergraph in which all hyperedges have the same
size. The reason for this step is that the paper \cite{FerroPotts}
shows how to reduce the latter to  the approximation of the Potts partition function of a \emph{graph},
which is the desired target of our reduction.

Let $q$ be a positive integer and let $\gamma$ be a positive  
efficiently approximable real.
We consider the following computational problem, which,  
like $\hPotts(q,\gamma)$,
is parameterised by~$q$ and~$\gamma$.
\begin{description}
\item[Problem] $\uhPotts(q,\gamma)$.
\item[Instance]  A uniform hypergraph $\hypergraph=(\hypervertices,\hyperedges)$.
\item[Output]  $\ZPotts(\hypergraph;q,\gamma)$.
\end{description} 
We will actually only use the following lemma with $\gamma=1$
but we state, and prove, the more general lemma, since it is no more difficult to
prove.

\begin{lemma}\label{lem:touniform}
Let $q$ be a positive integer and let $\gamma$ be a positive 
efficiently approximable real. Then
$$\hPotts(q,\gamma) \APred \uhPotts(q,\gamma).$$
\end{lemma}

\begin{proof}
 
Let $\hypergraph=(\hypervertices,\hyperedges)$ be an instance to 
$\hPotts(q,\gamma)$ with $|\hypervertices|=n$ and $|\hyperedges|=m$
and $\max(|\hyperedge| \mid \hyperedge \in \hyperedges)=t$.
Let~$s$ be any positive integer that is at least
$$ \frac{\log(4 q^{n+m(t-1)}{(1+\gamma)}^m)}
{\log(1+\gamma)}.$$ 
As with our other reductions, the exact value of~$s$ is not 
important, as long as it satisfies the above inequality, it is bounded from above by a
polynomial in~$n$ and~$m$,
and its can be computed in polynomial time (as a function of~$n$ and~$m$).
An appropriate~$s$ can be readily computed by computing crude upper and lower
bounds for~$\gamma$ and evaluating different values of~$s$ one-by-one to find one that is
sufficiently large, in terms of these bounds.

For every hyperedge $\hyperedge\in\hyperedges$,  
fix some vertex $v_\hyperedge\in\hyperedge$.
Introduce new vertices $\{u_{\hyperedge,i}\mid \hyperedge\in\hyperedges,i\in[t-1]\}$, and
let $\hypervertices' = \hypervertices \cup
\{u_{\hyperedge,i}\mid \hyperedge\in\hyperedges, i\in[t-1]\}$.
Let 
$$\hyperedges' = 
\Big\{
\hyperedge \cup \big\{u_{\hyperedge,i}\bigm| i\in[\, t-|\hyperedge|\, ]\big\} \Bigm|
\hyperedge\in\hyperedges
\Big\}
\cup 
\Big\{ 
\{v_\hyperedge,u_{\hyperedge,1},\ldots,u_{\hyperedge,t-1}\} \times [s] \Bigm|
\hyperedge \in \hyperedges
\Big\}.$$
That is, the multi-set $\hyperedges' $ has $s$ copies of the edge
$\{v_\hyperedge,u_{\hyperedge,1},\ldots,u_{\hyperedge,t-1}\} $ and
one copy of the edge $\hyperedge \cup \{u_{\hyperedge,i}\mid i\in[t-|\hyperedge|\,]\}$
for each hyperedge $\hyperedge\in \hyperedges$.
Let $\hypergraph' = (\hypervertices', \hyperedges')$. Note that
$\hypergraph'$ is $t$-uniform.

Now, the total contribution to $\ZPotts(\hypergraph';q,\gamma)$ from
configurations~$\sigma$ which are monochromatic on 
every edge $\{v_\hyperedge,u_{\hyperedge,1},\ldots,u_{\hyperedge,t-1}\}$ is exactly
 $\ZPotts(\hypergraph;q,\gamma) {(1+\gamma)}^{s m}$.
Also, the total contribution to $\ZPotts(\hypergraph';q,\gamma)$ from
any other configurations~$\sigma$
is at most $q^{n+m(t-1)} {(1+\gamma)}^{m} {(1+\gamma)}^{s(m-1)}$ since
there are at most $q^{n+m(t-1)}$ such configurations and $\gamma>0$.

So
\begin{align*} \ZPotts(\hypergraph;q,\gamma) \leq 
\frac{\ZPotts(\hypergraph';q,\gamma)}{{(1+\gamma)}^{s m}}
&\leq 
\ZPotts(\hypergraph;q,\gamma)  +
\frac{q^{n+m(t-1)} {(1+\gamma)}^{m}}
{ {(1+\gamma)}^s}\\
&\leq
\ZPotts(\hypergraph;q,\gamma)  + \frac14
\end{align*}
which completes the reduction.
\end{proof}

Finally, we are ready to put together the pieces to show that, for every
integer $q>2$,  
the problem of approximating the Potts partition function
is equivalent to a tree homomorphism problem.

\begin{theorem}
Let $q>2$ be a positive integer and let $\gamma$ be a 
positive efficiently approximable real.
Then
$\Potts(q,\gamma)\APeq \nHom{J_q}$.
\label{thm:junction}
\end{theorem}

\begin{proof}  
We start by establishing the reduction from $\nHom{J_q}$ to
$\Potts(q,\gamma)$.
By Lemmas \ref{lem:fromcol} and~\ref{lem:touniform}.
$$\nHom{J_q}\APred\hPotts(q,1)\APred\uhPotts(q,1).$$
To complete the sequence of reductions we need to know 
that the last problem is reducible to $\Potts(q,\gamma)$.  
Fortunately, this step already appears 
in the literature in a slightly different guise, so we just need to
explain how to translate the terminology from the earlier result 
to the current setting.  For every positive integer~$q$,
the partition function $\ZPotts(\hypergraph;q,\gamma)$ of the Potts model
on hypergraphs is
equal to the \emph{Tutte polynomial} $\ZTutte(\hypergraph;q,\gamma)$ 
  (whose definition
we will not need here).  
This equality is proved in \cite[Observation 2.1]{FerroPotts},
using the same basic line of argument that Fortuin and Kasteleyn~\cite{FK} used in
the graph case. 
Furthermore, for $q>2$,
Lemmas~9.1 and~10.1 of \cite{FerroPotts}
reduce the problem of approximating the Tutte partition function  
$\ZTutte(\hypergraph;q,1)$, where $\hypergraph$ is a \emph{uniform hypergraph},
to that of approximating the Tutte partition function  
$\ZTutte(G;q,\gamma)$, where $G$ is a \emph{graph}.
Given the equivalence between $\ZTutte(G;q,\gamma)$ and $\ZPotts(G;q,\gamma)$ 
mentioned earlier, we see that 
$$\uhPotts(q,1)\APred\Potts(q,\gamma),$$
completing the chain of reductions.

For the other direction, we will establish an AP-reduction from
$\Potts(q,\gamma)$ to the problem $\nHom{J_q}$.
To start, we note that since a graph is a special case of a uniform hypergraph,
Lemmas~9.1 and 10.1 of \cite{FerroPotts}
give an AP-reduction from
$\Potts(q,\gamma)$ to $\Potts(q,1)$.
(It is definitely not necessary to go via hypergraphs for this reduction,
but here it is easier to use the stated result than to repeat the work.)
Finally,  
Lemma~\ref{lem:tocol} shows that
$\Potts(q,1) \APred \nHom{J_q}$.
\end{proof}

\section{Inapproximability of counting tree homomorphisms}
\label{sec:hard}
Until now, it was not known whether or not a bipartite graph~$H$ exists for which
approximating $\nHom H$ is \SAT-hard.  It is perhaps surprising, then, 
to discover that $\nHom H$ may be \SAT-hard even when $H$ is a tree.  
However, the hardness 
result from Section~\ref{sec:weighted} provides a clue.  There it was shown 
that the weighted version $\wHom{H}$ is \SAT-hard whenever $H$ is a tree containing
$J_3$ as an induced subgraph.  If we were able to construct a tree~$H$, containing $J_3$,
that is able, at least in some limited sense, to simulate vertex weights, then 
we might obtain a reduction from $\wHom{J_3}$ to $\nHom{H}$.  That is roughly how we proceed in 
this section.  We will obtain our hard tree~$H$ by ``decorating'' the leaves of~$J_3$.
These decorations will match certain structures in the instance~$G$, so that particular
distinguished vertices in $G$ will preferentially be coloured with particular colours.  
Carrying through this idea requires~$H$ to have a certain level of complexity, and the 
tree~$\JS$ that we actually use (see Figure~\ref{fig:JS}) is about the smallest for which 
this approach works.  Presumably the same approach could also be applied starting at $J_q$, 
for $q>3$.  It is possible that there are trees~$H$
that are much smaller than~$\JS$   for which $\nHom{H}$ is \SAT-hard.
It is even possible that $\nHom{J_3}$ is \SAT-hard. But demonstrating this would require new ideas.   
    
Define vertex sets
\begin{align*}
X &=\{x_0,x_1\} \cup \{x_{2,i}\mid i\in[5]\} ,\\
Y &= \{y_0,y_1\} \cup \{y_{2,i}\mid i\in[4]\} \cup \{y_{3,i,j}\mid i\in[4],j\in[3]\}, \\
Z &= \{z_0,z_1\} \cup \{z_{2,i}\mid i\in[3]\} 
\cup \{z_{3,i,j}\mid i\in[3],j\in[3]\} 
\cup \{z_{4,i,j,k}\mid i\in[3],j\in[3],k\in[2]\},
\end{align*}
and edge sets
\begin{align*}
E_X &=\{(x_0,x_1)\} \cup \{(x_1,x_{2,i})\mid i\in[5]\} ,\\
E_Y &= \{(y_0,y_1)\} \cup \{(y_1,y_{2,i})\mid i\in[4]\} \cup \{(y_{2,i},y_{3,i,j})\mid i\in[4],j\in[3]\} ,\\
E_Z &= \{(z_0,z_1)\} \cup \{(z_1,z_{2,i})\mid i\in[3]\} 
\cup \{(z_{2,i},z_{3,i,j})\mid i\in[3],j\in[3]\} \\
&\qquad\null\cup \{ (
z_{3,i,j},z_{4,i,j,k})\mid i\in[3],j\in[3],k\in[2]\}.
\end{align*}
Let $\JS$ be the tree with vertex set $V(\JS)=\{w\} \cup X \cup Y \cup Z$
and edge set $$E(\JS)=\{(w,x_0),(w,y_0),(w,z_0)\} \cup E_X \cup E_Y \cup E_Z.$$
See Figure~\ref{fig:JS}.  
Consider the equivalence relation on  $V(\JS)$
defined by graph isomorphism --- two vertices of~$\JS$ are in the same equivalence
class if there is an isomorphism of~$\JS$ mapping one to the other.
The canonical representatives of the equivalence classes
are the vertices 
$w$, 
$x_0$, $x_1$, $x_{2,1}$,
$y_0$, $y_1$, $y_{2,1}$, $y_{3,1,1}$, 
$z_0$, $z_1$, 
$z_{2,1}$, $z_{3,1,1}$ and $z_{4,1,1,1}$. These are shown in the figure.

  \begin{figure}

\centering{
 \begin{tikzpicture}[fill=white,scale=0.7,
line width=0.5pt,inner sep=1pt,minimum size=2.5mm]
\pgfsetxvec{\pgfpoint{1.7cm}{0cm}}
\pgfsetyvec{\pgfpoint{0cm}{1.7cm}}
\path 
(0,0) node [draw,fill,circle,minimum size=0.6cm](w){ $w$}
(1,0) node [draw,fill,circle,minimum size=0.6cm](y0){ $y_0$}
(-1,0) node [draw,fill,circle,minimum size=0.6cm](x0){ $x_0$}
(0,-1) node [draw,fill,circle,minimum size=0.6cm](z0){ $z_0$}
(-2,0) node [draw,fill,circle,minimum size=0.6cm](x1){ $x_1$}
(-3,2) node [draw,fill,circle,minimum size=0.7cm](x21){ $x_{2,1}$}
(-3,1) node [draw,fill,circle,minimum size=0.7cm](x22){  }
(-3,0) node [draw,fill,circle,minimum size=0.7cm](x23){  }
(-3,-1) node [draw,fill,circle,minimum size=0.7cm](x24){  }
(-3,-2) node [draw,fill,circle,minimum size=0.7cm](x25){  }
(2,0) node [draw,fill,circle,minimum size=0.6cm](y1){ $y_1$}
(3,1.5) node [draw,fill,circle,minimum size=0.7cm](y21){ $y_{2,1}$}
(3,0.5) node [draw,fill,circle,minimum size=0.7cm](y22){  }
(3,-0.5) node [draw,fill,circle,minimum size=0.7cm](y23){  }
(3,-1.5) node [draw,fill,circle,minimum size=0.7cm](y24){  }
(5,3.5) node [draw,fill,circle,minimum size=0.5cm](y31){ $y_{3,1,1}$}
(5,2.5) node [draw,fill,circle,minimum size=0.5cm](y32){ }
(5,2) node [draw,fill,circle,minimum size=0.5cm](y33){ }
(5,1.5) node [draw,fill,circle,minimum size=0.5cm](y34){ }
(5,1) node [draw,fill,circle,minimum size=0.5cm](y35){ }
(5,0.5) node [draw,fill,circle,minimum size=0.5cm](y36){ }
(5,-0.5) node [draw,fill,circle,minimum size=0.5cm](y37){ }
(5,-1) node [draw,fill,circle,minimum size=0.5cm](y38){ }
(5,-1.5) node [draw,fill,circle,minimum size=0.5cm](y39){ }
(5,-2) node [draw,fill,circle,minimum size=0.5cm](y310){ }
(5,-2.5) node [draw,fill,circle,minimum size=0.5cm](y311){ }
(5,-3) node [draw,fill,circle,minimum size=0.5cm](y312){ }
;
\tikzstyle{level 1}=[sibling distance=70mm]
\tikzstyle{level 2}=[sibling distance=23mm]
\tikzstyle{level 3}=[sibling distance=15mm]
\draw (0,-2) node [draw,fill,circle,minimum size=0.6cm](z1){ $z_1$}
    child {node [draw,fill,circle,minimum size=0.6cm] {$ z_{2,1}$  }
                child {node [draw,fill,circle,minimum size=0.6cm] { $  z_{3,1,1}$ }
                             child {node [draw,fill,circle,minimum size=0.5cm] { $  z_{4,1,1,1}$ }}
                             child {node [draw,fill,circle,minimum size=0.5cm] {  }}
                                   }
                child {node [draw,fill,circle,minimum size=0.6cm] {  } 
                           child {node [draw,fill,circle,minimum size=0.5cm] { }}
                             child {node [draw,fill,circle,minimum size=0.5cm] {  }}  
                             }
                child {node [draw,fill,circle,minimum size=0.6cm] {  }
                             child {node [draw,fill,circle,minimum size=0.5cm] { }}
                             child {node [draw,fill,circle,minimum size=0.5cm] {  }}     
                }
             }
    child {node [draw,fill,circle,minimum size=0.6cm] {  }
                child {node [draw,fill,circle,minimum size=0.6cm] {  }
                             child {node [draw,fill,circle,minimum size=0.5cm] { }}
                             child {node [draw,fill,circle,minimum size=0.5cm] {  }}                     
   }
                child {node [draw,fill,circle,minimum size=0.6cm] {  }
                             child {node [draw,fill,circle,minimum size=0.5cm] { }}
                             child {node [draw,fill,circle,minimum size=0.5cm] {  }}                  
                }
                child {node [draw,fill,circle,minimum size=0.6cm] {  }
                             child {node [draw,fill,circle,minimum size=0.5cm] { }}
                             child {node [draw,fill,circle,minimum size=0.5cm] {  }}                   
   }
    }
    child {node [draw,fill,circle,minimum size=0.6cm] { }
                child {node [draw,fill,circle,minimum size=0.6cm] {  }
                             child {node [draw,fill,circle,minimum size=0.5cm] { }}
                             child {node [draw,fill,circle,minimum size=0.5cm] {  }}                    
                }
                child {node [draw,fill,circle,minimum size=0.6cm] {  }
                             child {node [draw,fill,circle,minimum size=0.5cm] { }}
                             child {node [draw,fill,circle,minimum size=0.5cm] {  }}         
                }
                child {node [draw,fill,circle,minimum size=0.6cm] {  }
                             child {node [draw,fill,circle,minimum size=0.5cm] { }}
                             child {node [draw,fill,circle,minimum size=0.5cm] {  }}     
                }
    }
     ;
 
\draw [-] (x1)--(x0) -- (w) -- (y0) -- (y1);
\draw [-] (w)--(z0)--(z1); 
\draw [-] (x21)--(x1)--(x22);
\draw [-] (x23)--(x1)--(x24);
\draw [-] (x1)--(x25);
\draw [-] (y21)--(y1)--(y22);
\draw [-] (y23)--(y1)--(y24);
\draw [-] (y31)--(y21)--(y32);
\draw [-] (y21)--(y33); 
\draw [-] (y34)--(y22)--(y35);
\draw [-] (y22)--(y36);  
\draw [-] (y37)--(y23)--(y38);
\draw [-] (y23)--(y39);  
\draw [-] (y310)--(y24)--(y311);
\draw [-] (y24)--(y312);   
 
\end{tikzpicture}

 }
 \caption{The tree $\JS$.}
\label{fig:JS}
\end{figure}  

In this section, we will  show that $\SAT$
is AP-reducible to~$\nHom{\JS}$.
We start by   identifying  relevant structure in~$\JS$.

A simple path in a graph is a path in which no vertices are repeated.
For every vertex~$h$ of~$\JS$, 
and every positive integer~$k$, let $d_k(h)$ be the
number of simple length-$k$ paths from~$h$.
The values $d_1(h)$, $d_2(h)$ and $d_3(h)$ can be calculated
for each canonical representative $h\in V(\JS)$ by inspecting the
definition of~$\JS$ (or its drawing 
in Figure~\ref{fig:JS}). These values are recorded in 
the first four columns of the table in Figure~\ref{JStable}. 
\begin{figure}
   $$
 \begin{array}{|c||c|c|c|c|c|c|}
 \hline
  h & d_1(h) & d_2(h) & d_3(h)&w_1(h)&w_2(h)&w_3(h)\\
 \hline
 w           & 3  & 3 & 12&3&6&24\\
 x_0       & 2 & 7 & 2&2&9&13\\
 x_1       & 6 & 1 & 2&{\bf 6}&7&39\\
 x_{2,1}   & 1 & 5 & 1&1&6&7\\ 
 y_0       & 2 & 6 & 14&2&8&24\\ 
  y_1      & 5 & 13 & 2&5&{\bf 18}&40\\
 y_{2,1}   & 4 & 4 & 10&4&8&30\\
 y_{3,1,1} & 1 & 3 & 4&1&4&8\\ 
 z_0        & 2 & 5 & 11&2&7&20\\
 z_1         & 4  & 10 & 20&4&14&{\bf 46}\\
 z_{2,1}     & 4 & 9 & 7&4&13&32\\
 z_{3,1,1}   & 3 & 3 & 7&3&6&19\\
 z_{4,1,1,1} & 1 & 2 & 3&1&3&6\\
 \hline
  \end{array}
  $$ 
  \caption{ For each canonical representative $h\in V(\JS)$, we record the
  values of $w_1(h)=d_1(h)$, $w_2(h)=d_1(h)+d_2(h)$ and $w_3(h)=d_1^2(h)+d_2(h)+d_3(h)$.}
  \label{JStable}
  \end{figure}
   
Now let $w_k(h)$ denote the number of length-$k$ walks from~$h$ in~$\JS$.
Clearly, $w_1(h)=d_1(h)$ since $\JS$ has no self-loops,  so all length-$1$ walks
are simple paths.
Next, note that $w_2(h)=d_1(h) + d_2(h)$.
To see this, note that
every length-$2$ walk from~$h$ is either a simple length-$2$ path from~$\JS$,
or it is a walk obtained by taking an edge from~$h$, and then going back to~$h$.
Finally, $w_3(h) = d_1(h)^2 + d_2(h)+d_3(h)$ 
since every length-$3$ walk from~$h$ is
one of the following:
\begin{itemize}
\item a simple length-$3$ path from~$h$,
\item a simple length-$2$ path from~$h$, with the last edge repeated in reverse, or
\item a simple length-$1$ path from~$h$ with the last edge repeated in reverse,
followed by another simple length-$1$ path from~$h$.
\end{itemize}
These values are recorded, for each canonical representative $h\in V(\JS)$,
in the last three columns of the table in Figure~\ref{JStable}.
The important fact that we will use is that $w_1(h)$ is uniquely maximised at $h=x_1$,
  $w_2(h)$ is uniquely maximised at $h=y_1$, and $w_3(h)$ is uniquely maximised at $h=z_1$.
  (These are shown in boldface in the table.)
  
We are now ready to prove the following theorem.  
   
\begin{theorem}
\label{thm:hardH}
$\SAT \APred \nHom{\JS}$.
\end{theorem} 
  
\begin{proof} 
 
 By Lemma~\ref{lem:cut},
it suffices to  give an AP-reduction from
\MultiCutCount{$3$} to $\nHom{\JS}$. 
The basic construction 
follows the outline of
the reduction 
developed in the proof of  Lemma~\ref{lem:hardweighted}.
However, unlike the situation of Lemma~\ref{lem:hardweighted},
the target problem $\nHom{\JS}$ does
not include weights, so we must develop gadgetry to simulate the
role of these.

Let $b$, $G=(V,E)$, $\alpha$, $\beta$ and $\gamma$ be an input to \MultiCutCount{$3$}.
Let  
$s= 3 + |E(G)|+2|V(G)|$.
(As before, the exact size of~$s$ is not important, but
it has to be at least this big to make the calculation work,
and it has to be at most a polynomial in the size of~$G$.)

Let $G'$ be the graph defined in the proof of Lemma~\ref{lem:hardweighted}.
In particular,
let $V'(G)= \{(e,i) \mid e\in E(G),i\in[s]\}$.
Then let $G'$
be the graph with vertex set
$V(G') = V(G) \cup V'(G)$
and edge set 
 $$E(G') = \{(u,(e,i)) \mid u\in V(G), (e,i)\in V'(G),  
\mbox{and $u$ is  an endpoint of~$e$}  \}.$$

Now let $r$ be
any positive integer
such that
\begin{equation}
\label{eq:r}
{\left(
\frac{46}{40}
\right)}^r \geq 8 {|V(\JS)|}^{|V(G)|+ s |E(G)| + 7}.
\end{equation}
For concreteness, take $r$ to be the smallest integer satisfying
(\ref{eq:r}). Once again, the exact value of~$r$ is not so important.
Any~$r$ would work as long as it is at most a polynomial in the size of~$G$,
and it satisfies (\ref{eq:r}).

We will construct an instance~$G''$ of $\nHom{\JS}$
by adding some gadgets to~$G'$.
First, we define the gadgets.
\begin{itemize}
\item
Let $\Gamma_{x}$ be a 
graph with vertex set
$V(\Gamma_{x}) = 
\{ v_{x_1} \} \cup \bigcup_{i\in[r]} \{v_{x,i}\}$
and edge set $E(\Gamma_x) = 
\bigcup_{i\in [r]} \{(v_{x_1},v_{x,i})\}$. 
\item 
Let $\Gamma_y$ be a graph
with vertex set
$V(\Gamma_y) = \{v_{y_1}\} \cup 
\bigcup_{i\in[r]} 
\{v_{y,i},v'_{y,i}\} $
and edge set 
$E(\Gamma_y) = \bigcup_{i\in [r]} \{(v_{y_1},v_{y,i}),(v_{y,i},v'_{y,i})
\}$.
\item 
Let $\Gamma_z$ be a graph
with vertex set
$V(\Gamma_z) = \{v_{z_1}\} \cup 
\bigcup_{i\in[r]} 
\{ v_{z,i}, v'_{z,i}, v''_{z,i} \} $
and edge set 
$E(\Gamma_x) = \bigcup_{i\in [r]} \{ (v_{z_1},v_{z,i}),(v_{z,i},v'_{z,i}),(v'_{z,i},v''_{z,i})
\}$.
\end{itemize}
Finally, let 
$$V(G'') = V(G') \cup \{v_w,v_{x_0},v_{y_0},v_{z_0}\} \cup V(\Gamma_x) \cup
V(\Gamma_y) \cup V(\Gamma_z),$$
and
\begin{align*} 
E(G'') &=  
\{
(v_w,v_{x_0}),(v_w,v_{y_0}),(v_w,v_{z_0}),
(v_{x_0},v_{x_1}),(v_{y_0},v_{y_1}),(v_{z_0},v_{z_1}),
(v_{x_1},\alpha),(v_{y_1},\beta),(v_{z_1},\gamma)
\} \\
& \cup
E(G') \cup
\{(v_w,v) \mid v\in V(G)\} 
\cup  
E(\Gamma_x) \cup E(\Gamma_y) \cup E(\Gamma_z).
\end{align*}
A picture of the instance $G''$ is shown in Figure~\ref{fig:GInstance}.

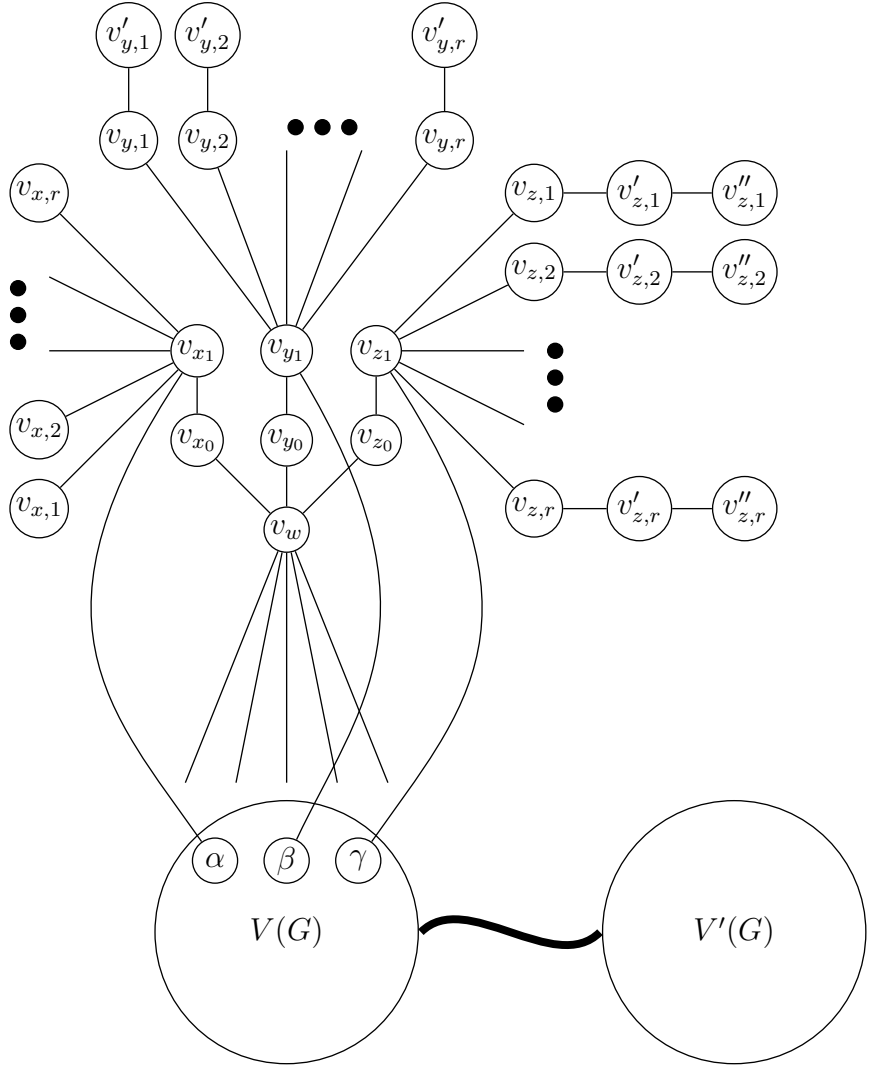
\begin{figure}
\centering{
 \begin{tikzpicture}[fill=white,scale=0.7,
line width=0.5pt,inner sep=1pt,minimum size=2.5mm]
\pgfsetxvec{\pgfpoint{1.7cm}{0cm}}
\pgfsetyvec{\pgfpoint{0cm}{1.7cm}}
\path 
(0,0) node [draw,fill,circle,minimum size=0.6cm](w){ $v_w$}
[grow'=down, level distance=50mm,sibling distance=10mm]
child {node {}}
child {node {}}
child {node {}}
child {node {}}
child {node {}}
(-1,1) node [draw,fill,circle,minimum size=0.6cm](x0){ $v_{x_0}$}
(0,1) node [draw,fill,circle,minimum size=0.6cm](y0){ $v_{y_0}$}
(1,1) node [draw,fill,circle,minimum size=0.6cm](z0){ $v_{z_0}$}
(-1,2) node [draw,fill,circle,minimum size=0.6cm](x1){ $v_{x_1}$} 
[grow'=left, level distance=30mm,sibling distance=15mm]
   child {node [draw,fill,circle,minimum size=0.6cm] { $  v_{x,1}$} }  
   child {node [draw,fill,circle,minimum size=0.6cm] { $  v_{x,2}$} }  
   child {node  { } }  
   child {node  { } }  
   child {node [draw,fill,circle,minimum size=0.6cm] { $  v_{x,r}$} }  
(0,2) node [draw,fill,circle,minimum size=0.6cm](y1){ $v_{y_1}$}
[grow'=up, level distance=40mm]
  child {node [draw,fill,circle,minimum size=0.6cm] { $  v_{y,1}$}  
  [level distance=20mm]
  child {node [draw,fill,circle,minimum size=0.6cm] { $  v'_{y,1}$}}
  } 
  child {node [draw,fill,circle,minimum size=0.6cm] { $  v_{y,2}$}  
  [level distance=20mm]
  child {node [draw,fill,circle,minimum size=0.6cm] { $  v'_{y,2}$}}
  }   
   child {node  { } }  
   child {node  { } }    
 child {node [draw,fill,circle,minimum size=0.6cm] { $  v_{y,r}$}  
  [level distance=20mm]
  child {node [draw,fill,circle,minimum size=0.6cm] { $  v'_{y,r}$}}
  }  
(1,2) node [draw,fill,circle,minimum size=0.6cm](z1){ $v_{z_1}$}
 [grow'=right, level distance=30mm]
   child {
        node [draw,fill,circle,minimum size=0.6cm] { $  v_{z,1}$} [level distance=20mm]
            child { node [draw,fill,circle,minimum size=0.6cm] { $  v'_{z,1}$} 
                child { node [draw,fill,circle,minimum size=0.6cm] { $  v''_{z,1}$} 
                }
            }
   }  
   child {
        node [draw,fill,circle,minimum size=0.6cm] { $  v_{z,2}$} [level distance=20mm]
            child { node [draw,fill,circle,minimum size=0.6cm] { $  v'_{z,2}$} 
                child { node [draw,fill,circle,minimum size=0.6cm] { $  v''_{z,2}$} 
                }
            }
   }  
   child {node  { } }  
   child {node  { } }  
child {
        node [draw,fill,circle,minimum size=0.6cm] { $  v_{z,r}$} [level distance=20mm]
            child { node [draw,fill,circle,minimum size=0.6cm] { $  v'_{z,r}$} 
                child { node [draw,fill,circle,minimum size=0.6cm] { $  v''_{z,r}$} 
                }
            }
   }  
;
\draw [-] (x1) -- (x0) --(w) -- (y0) -- (y1);
\draw [-] (w) -- (z0) -- (z1);
\path (-3,2.1) node [draw,fill=black,circle,minimum size=0.2cm] {}; 
\path (-3,2.4) node [draw,fill=black,circle,minimum size=0.2cm] {}; 
\path (-3,2.7) node [draw,fill=black,circle,minimum size=0.2cm] {}; 
\path (3,1.4) node [draw,fill=black,circle,minimum size=0.2cm] {}; 
\path (3,1.7) node [draw,fill=black,circle,minimum size=0.2cm] {}; 
\path (3,2.0) node [draw,fill=black,circle,minimum size=0.2cm] {}; 
\path (0.1,4.5) node [draw,fill=black,circle,minimum size=0.2cm] {}; 
\path (0.4,4.5) node [draw,fill=black,circle,minimum size=0.2cm] {}; 
\path (0.7,4.5) node [draw,fill=black,circle,minimum size=0.2cm] {}; 
\path (0,-4.5) node   {$V(G)$};
\draw (0,-4.5) circle (2.5cm);
\draw (5,-4.5) circle (2.5cm);
\path (5,-4.5) node  {$V'(G)$};
\draw [line width=3pt]  (1.5,-4.5) .. controls (2,-4)  and (3,-5) .. (3.5,-4.5);
\path 
(-0.8,-3.7) node [draw,fill,circle,minimum size=0.6cm](alpha) { $\alpha$}
(0,-3.7) node  [draw,fill,circle,minimum size=0.6cm](beta){ $\beta$ }
(0.8,-3.7) node  [draw,fill,circle,minimum size=0.6cm](gamma){ $\gamma$}
;
\draw (x1) .. controls (-3,-1) and (-2,-2) .. (alpha);
\draw (z1) .. controls (3,-1) and (2,-2) .. (gamma);
\draw (y1) .. controls (1.8,-1) and (0.8,-2) .. (beta);
\end{tikzpicture}

 }
 \caption{The instance~$G''$.  The thick curved line between $V(G)$ and
 $V'(G)$ indicates that  the edges in~$E(G')$ go between vertices in~$V(G)$ and
 vertices in~$V'(G)$, but these are not shown.  
 Vertex $v_w$ is connected to each vertex in~$V(G)$.
 }
\label{fig:GInstance}

\end{figure}

We say that a homomorphism~$\sigma$ from~$G''$ to~$\JS$ is  \emph{typical}
if $\sigma(v_{x_1})=x_1$,
$\sigma(v_{y_1})=y_1$, 
and $\sigma(v_{z_1})=z_1$.
Note that, in a typical homomorphism, $\sigma(v_w)=w$,
so $\sigma(V(G))=\{x_0,y_0,z_0\}$ 
and $\sigma(V'(G)) \subseteq \{w,x_1,y_1,z_1\}$.
Also, $\sigma(\alpha)=x_0$,
$\sigma(\beta)=y_0$, and $\sigma(\gamma)=z_0$.

If $\sigma$ is a typical homomorphism, then let 
\begin{align*}
\bichrom{\sigma} =
\{ e \in E(G) \mid \quad & 
\mbox{the vertices of~$V(G)$ corresponding to } \\
& \mbox{the endpoints of~$e$ are mapped to different colours by~$\sigma$} 
\}.\end{align*}
Note that, for every typical homomorphism~$\sigma$,
 $\bichrom{\sigma}$ is a multiterminal cut 
for the graph~$G$ with terminals~$\alpha$, $\beta$ and~$\gamma$.

For every multiterminal cut 
$E'$ of~$G$,
let $\components{E'}$
denote the number of components
in the graph $(V,E\setminus E')$.
For each multiterminal cut~$E'$,
let $Z_{E'}$ denote
the number of typical homomorphisms~$\sigma$ from~$G''$ to~$\JS$
such that $\bichrom{\sigma} = E'$.

As in the proof of Lemma~\ref{lem:hardweighted}, $\components{E'}\geq 3$.
If $\components{E'}=3$ then
$$Z_{E'} =  2^{s|E(G)-E'|} 6^r 18^r 46^r
= 2^{s|E(G)-E'|} 4968^r.$$
The $2^{s|E(G)-E'|}$
comes from the two choices for the colour of each vertex $(e,i)$ with $e\in E(G)-E'$, as before.
The $6^r$ comes from the choices for the vertices in 
$V(\Gamma_x)\setminus \{x_1\}$
according to column~5 of the table in Figure~\ref{JStable}.
The $18^r$ comes from the choices for the vertices in $
V(\Gamma_y)\setminus \{y_1\}$
(in column~6)
and the $46^r$ comes from the choices for the vertices in $
V(\Gamma_z)\setminus \{z_1\}$
(in column~7).

Also, for any multiterminal cut 
$E'$ of~$G$,
$$Z_{E'} \leq   2^{s|E(G)-E'|} 3^{\components{E'}-3}  4968^r,$$
since
in any typical homomorphism~$\sigma$,
    the component
of~$\alpha$ is mapped to~$x_0$ by~$\sigma$,
the component of~$\beta$ is mapped to~$y_0$, 
the component of~$\gamma$ is mapped to~$z_0$,  and each remaining
component  is mapped to a colour in $\{x_0,y_0,z_0\}$.

Let $Z^*= 2^{s|E(G)-b|} 4968^r$.
If $E'$ has size~$b$ then 
$\components{E'}=3$. (Otherwise, there would be a smaller multiterminal cut,
contrary to the definition of \MultiCutCount{$3$}.)
So, in this case, 
\begin{equation}
Z_{E'} = Z^*.
\label{eq:goodcuts}
\end{equation}

If $E'$ has size $b'>b$ then 
$$Z_{E'} \leq 2^{s|E(G)-b'|} 3^{\components{E'}-3}  4968^r
               = 2^{-s(b'-b)} 3^{\components{E'}-3} Z^*
               \leq 2^{-s} 3^{|V(G)|} Z^*.
$$
Clearly, there are at most $2^{|E(G)|}$   multiterminal cuts~$E'$.
So, using the definition of~$s$,  
\begin{equation}
\label{eq:bigcuts}
\sum_{E' : |E'|>b} Z_{E'} \leq \frac{Z^*}{8}.
\end{equation}

Now let $Z^-$ denote the number of homomorphisms from~$G''$ to~$\JS$
that are not typical. 
Now
$$ 
Z^- \leq |V(\JS)|^{|V(G)|+|V'(G)|+7 } {(40/46)}^r 4968^r,
$$
since there are at most $|V(\JS)|$ colours for each of the vertices in
$$V(G)\cup V'(G) \cup \{v_w,v_{x_0},v_{y_0},v_{z_0},v_{x_1},v_{y_1},v_{z_1}\}.$$
Also, given that the assignment to $v_{x_1}$, $v_{y_1}$ and $v_{z_1}$ is
not precisely $x_1$, $y_1$ and $z_1$, respectively,
it can be seen from the table in Figure~\ref{JStable} 
that the number of possibilities for the remaining vertices is
at most $(40/46)^r$ times as large as it would otherwise have been.
(For example, from the last column of the table, colouring
$v_{z_1}$ with $y_1$ instead of with~$z_1$ would give exactly 
$40^r$ choices for the colours of the vertices in $
\Gamma_z \setminus \{v_{z_1}\}$ instead of 
$46^r$ choices. The differences in the other columns are more substantial than this.)
Since $|V'(G)|=s |E(G)|$,
$$Z^- \leq 
{|V(\JS)|}^{|V(G)|+s|E(G)|+7} {(40/46)}^r 4968^r.$$
We can assume that $b\leq |E(G)|$ (otherwise, the number of size-$b$ multiterminal cuts
is trivially~$0$)
so from the definition of~$Z^*$,
$$
 Z^- \leq 
{|V(\JS)|}^{|V(G)|+s|E(G)|+7} {(40/46)}^r  Z^*.
$$
Using Equation~(\ref{eq:r}), we get
\begin{equation}
\label{eq:nocut}
Z^- \leq \frac{Z^*}{8}.
\end{equation}

From Equation~(\ref{eq:goodcuts}), we find that,
if  there are $N$ size-$b$ multiterminal cuts
then  
$$Z_{\JS}(G) = N Z^* + \sum_{E' : |E'|>b} Z_{E'} + Z^-.$$
So applying 
Equations  (\ref{eq:bigcuts}) and (\ref{eq:nocut}), 
we get
$$ N \leq \frac{Z_{\JS}(G)}{Z^*} \leq N + \frac{1}{4}.$$

Thus, we have an AP-reduction
from \MultiCutCount{$3$}  to $\nHom{\JS}$.
To determine the accuracy with which~$Z(G)$ should be approximated
in order to achieve a given accuracy in the approximation to~$N$,
see the   proof of Theorem 3 of \cite{APred}.
\end{proof}

\section{The Potts partition function and proper colourings of bipartite graphs}
\label{sec:bqcol}

Let $q$ be any integer greater than~$2$. Consider the following computational problem.

\begin{description}
\item[Problem] $\bqcol q$.
\item[Instance]  A bipartite graph $G$.
\item[Output]  The number of proper $q$-colourings of $G$.
\end{description} 

Dyer et al.~\cite[Theorem 13]{APred} showed that $\BIS \APred \bqcol q$.
However, it may be the case that $\bqcol q$ is easier to approximate than $\SAT$. Certainly,
no AP-reduction from $\SAT$ to $\bqcol q$ has been discovered (despite some effort!).
Therefore, it seems worth recording the following
upper bound on the complexity of $\nHom{J_q}$, which is an easy consequence of
Theorem~\ref{thm:junction}.

\begin{corollary}\label{cor:bqcol}
Let $q>2$ be a positive integer. Then
$\nHom{J_q} \APred \bqcol q$.
\end{corollary}

Corollary~\ref{cor:bqcol} follows immediately from 
Lemma~\ref{lem:bqcol} below by applying
Theorem~\ref{thm:junction} with
$\gamma=1/(q-2)$.

\begin{lemma}\label{lem:bqcol}
Let $q>2$ be a positive integer. Then
$\Potts(q,1/(q-2)) \APred \bqcol q$.
\end{lemma}

\begin{proof} 
Let $\graph=(\graphvertices,\graphedges)$ be an input to $\Potts(q,1/(q-2))$.
Let $\graph'$ be the two-stretch of $\graph$ constructed as in the proof of Lemma~\ref{lem:tocol}.
In particular, 
$G'$
is the bipartite
graph with 
$$V(G') = V(G) \cup E(G)$$
and 
$$E(G') = \{(u,e) \mid u\in V(G), e \in E(G),  
\mbox{and $u$ is  an endpoint of~$e$}  \}.$$

Consider an assignment $\sigma\colon V(G) \to [q]$ and
an edge $e=(u,v)$ of $\graph$.
If $\sigma(u)\neq \sigma(v)$ then 
there are $q-2$ ways to 
colour the midpoint vertex corresponding to~$e$
so that it receives a different colour from~$\sigma(u)$ and~$\sigma(v)$.
However, if $\sigma(u)=\sigma(v)$ then there are $q-1$ possible colours for the midpoint vertex.
 
 Let $N$ denote the number of proper $q$-colourings of~$G'$. Then
since $(q-1)/(q-2)-1=1/(q-2)$, we have
$$
N = {(q-2)}^{|\graphedges|}
\sum_{\sigma:\graphvertices\rightarrow[q]} 
{\left(\frac{q-1}{q-2}\right)}^{\mathrm{mono}(\sigma)}
= {(q-2)}^{|\graphedges|}
\ZPotts(\graph;q,  1/(q-2)),$$
where 
$\mathrm{mono}(\sigma)$
is the number of edges $e\in E(G)$ whose endpoints in $V(G)$ are mapped to 
the same colour by~$\sigma$.

\end{proof}
    
\section{The Potts partition function and the weight enumerator of a code}
\label{sec:we}    

A {\it linear code\/} $C$ of length $N$ over a finite field $\Fq$ is a linear subspace of 
$\Fq^N$.  
If the subspace has dimension~$r$ then the code may be
specified by an $r\times N$ {\it generating matrix}~$M$ over~$\Fq$ 
whose rows form a basis for the code.  
For any real number $\lambda$, the weight enumerator of the code is given
by $W_M(\lambda)=\sum_{w\in C}\lambda^{\|w\|}$ where $\|w\|$ 
is the number of non-zero entries in~$w$.  ($\|w\|$ is usually called the {\it Hamming
weight\/} of~$w$.)  We
consider the following computational problem, parameterised by $q$ and~$\lambda$.
\begin{description}
\item[Problem] $\WE q\lambda$.
\item[Instance]  A generating matrix $M$ over $\Fq$.
\item[Output]  $W_M(\lambda)$.
\end{description} 
In \cite{WeightEnum}, the authors considered the special case $q=2$ and obtained various 
results on the complexity of $\WE 2\lambda$, depending on $\lambda$.  Here we 
show that, 
for any prime~$p$,
$\WE p\lambda$ provides an upper bound on the complexity of $\Potts(p^k,\gamma)$.

\begin{theorem}\label{thm:PottsToWE}
Suppose that $p$ is a prime, $k$ is a positive integer 
satisfying $p^k>2$
and $\lambda\in(0,1)$ is
an efficiently computable real.  Then
$$\Potts(p^k,1)\APred\WE p\lambda.$$
\end{theorem}

The following corollary follows immediately from Theorem~\ref{thm:PottsToWE} and
Theorem~\ref{thm:junction}.
\begin{corollary}
\label{newcor}
Suppose that $p$ is a prime, 
$k$ is a positive integer 
satisfying $p^k>2$
and $\lambda\in(0,1)$ is
an efficiently computable real. 
Then
$\nHom{J_{p^k}} \APred
\WE p\lambda$. 
\end{corollary}
 
The condition $p^k>2$ can in fact be removed from Corollary~\ref{newcor},
even though the result does not follow from Theorem~\ref{thm:PottsToWE} 
in this situation.  
For the missing case where $p=2$ and $k=1$, 
Lemma~\ref{lem:intermediate} gives
$\nHom{J_{2}} \APred \BIS$ and \cite[Cor.~7, Part~(4)]{WeightEnum}
show $\BIS \APred \WE {2}{\lambda}$.
A striking feature of Corollary~\ref{newcor} is that it provides a uniform upper bound
on the complexity of the infinite sequence of problems $\nHom{J_{p^k}}$, with $p$ fixed
and $k$ varying.
This uniform upper bound is interesting if 
(as we suspect)
$\WE p\lambda$ is not itself equivalent to \SAT{} via AP-reducibility.

\begin{proof}[Proof of Theorem~\ref{thm:PottsToWE}]
Let 
$q=p^k$ and let
$\gamma=\lambda^{-q(p-1)/p}-1>0$.  
Since Theorem~\ref{thm:junction}
shows $\Potts(p^k,1)\APeq \nHom{J_{p^k}}\APeq\Potts(p^k,\gamma)$,
it is
enough to 
given an AP-reduction from
$\Potts(p^k,\gamma)$ to $\WE p\lambda$.
So suppose 
$G=(V,E)$
is a graph with $n$ vertices and $m$ edges.
We wish to evaluate 
\begin{equation}\label{eq:PottsDef}
\ZPotts(G;q,\gamma)=\sum_{\sigma:V\to[q]}(1+\gamma)^{\mono(\sigma)}.
\end{equation}
Our aim is to construct an instance of the weight enumerator problem 
whose solution is the above expression, modulo an easily 
computable factor.
Introduce a collection of variables $X=\{x^v_i\mid v\in V \text{ and }i\in[k]\}$.  
To each assignment $\sigma:V\to[q]$
we define an associated assignment $\sigmahat:X\to\Fp$ as follows:  
for all $v\in V$,
$$
\big(\sigmahat(x_1^v),\sigmahat(x_2^v), \ldots,\sigmahat(x_k^v)\big)=\phi(\sigma(v)),
$$
where $\phi$ is any fixed bijection $[q]\to 
\Fp^k$.
Note that $\sigma\mapsto\sigmahat$ is a bijection from assignments $V\to[q]$ to assignments $X\to\Fp$.
(Informally, we have coded the spin at each vertex as a $k$-tuple 
of variables taking values in~$\Fp$.)
   
Let $\ell_1(z_1,\ldots,z_k),\ldots,\ell_q(z_1,\ldots,z_k)$ be an enumeration of all linear forms 
$\alpha_1z_1+\alpha_2z_2+\cdots+\alpha_kz_k$ over $\Fp$, 
where $(\alpha_1,\alpha_2,\ldots,\alpha_k)$ ranges over $\Fp^k$.
This collection of linear forms has the following property:
\begin{equation}\label{eq:prop}
\begin{split}
&\text{If $z_1=z_2=\cdots z_k=0$, then all of 
$\ell_1(z_1,\ldots,z_k),\ldots,\ell_q(z_1,\ldots,z_k)$ are zero;}\\
&\text{otherwise, precisely $q/p=p^{k-1}$ of 
$\ell_1(z_1,\ldots,z_k),\ldots,\ell_q(z_1,\ldots,z_k)$ are zero.}
\end{split}
\end{equation}
The first claim in~(\ref{eq:prop}) is trivial.  
To see the second, assume without loss of generality
that $z_1\not=0$.  Then, for any choice of $(\alpha_2,\ldots,\alpha_k)\in\Fp^{k-1}$, there is 
precisely one choice for $\alpha_1\in\Fp$ that makes $\alpha_1z_1+\cdots+\alpha_kz_k=0$.

Now 
give an arbitrary direction to each edge $(u,v)\in E$ and
consider the system $\Lambda$ of linear equations 
$$
\Big\{
\ell_j
\big(\sigmahat(x^v_1)-\sigmahat(x^u_1),\,\sigmahat(x^v_2)-\sigmahat(x^u_2),\,
   \ldots,\,\sigmahat(x^v_k)-\sigmahat(x^u_k)\big)=0:
   j\in[q]
   \text{ and } 
   (u,v)
   \in E\Big\}.
$$
(We view $\Lambda$ as a multiset, so the trivial equation $0=0$ 
arising from the linear form $\ell_j$ with $\alpha_1=\alpha_2 = \cdots = \alpha_k=0$
occurs $m$ times,
a convention that makes the following calculation simpler.)
Denote by $\sat(\sigmahat)$ the number of satisfied equations in $\Lambda$.  
Then, from~(\ref{eq:prop}), 
$$\sat(\sigmahat)=q\mono(\sigma)+\frac qp(m-\mono(\sigma)),$$ 
and hence
$$
\mono(\sigma)=\frac p{(p-1)q}\sat(\sigmahat)-\frac m{p-1}.
$$
Noting that $1+\gamma=\lambda^{-q(p-1)/p}$,
\begin{align}
\sum_{\sigma:V\to[q]}(1+\gamma)^{\mono(\sigma)}
&=\sum_{\sigmahat:
X
\to\Fp}(1+\gamma)^{(p/(p-1)q)\sat(\sigmahat)-m/(p-1)}\notag\\
&=
\lambda^{qm/p}
\sum_{\sigmahat:
X
\to\Fp}\lambda^{-\sat(\sigmahat)}\notag\\
&=
\lambda^{-(1-1/p)qm}
\sum_{\sigmahat:
X
\to\Fp}\lambda^{\unsat(\sigmahat)},\label{eq:unsat}
\end{align}
where $\unsat(\sigmahat)=qm-\sat(\sigmahat)$ is the number of unsatisfied equations in $\Lambda$.

The system $\Lambda$ has $qm$ equations in $kn$ variables, so we may write it 
in matrix form $A\bsigma=\mathbf0$, where $A$ is a $(qm\times kn)$-matrix,
and $\bsigma$ is a $kn$-vector over~$\Fp$.  The columns of $A$ and the 
components of $\bsigma$ are indexed by pairs $(i,v)\in[k]\times V$, and the 
$(i,v)$-component of $\bsigma$ is $\sigmahat(x_i^v)$.  Enumerating the columns
of $A$ as $\bfa_i^v\in\Fp^{qm}$ for $(i,v)\in[k]\times V$, we may re-express $\Lambda$
in the form
$$
\sum_{i\in[k],v\in V}\sigmahat(x_i^v)\,\bfa_i^v=\mathbf0,
$$
where $\mathbf0$ is the length-$qm$ zero vector. 
Then $\unsat(\sigmahat)$ is the Hamming weight of the 
length-$qm$
vector
$\bfb(\sigmahat)=\sum_{i,v}\sigmahat(x_i^v)\,\bfa_i^v$.  
As $\sigmahat$ ranges over all assignments 
$X\to\Fp$, so $\bfb(\sigmahat)$ ranges
over the vector space (or code) 
$$C=\Big\{\sum_{i,v}\sigmahat(x_i^v)\,\bfa_i^v\Bigm| \sigmahat:
X\to\Fp\Big\}
=\langle \bfa_i^v\mid i\in[k],v\in V\rangle$$  
generated by the vectors $\{\bfa_i^v\}$.

We will argue that the mapping sending $\sigmahat$ to $\bfb(\sigmahat)$ is $q$ to~1,
from which it follows that $\sum_{\sigmahat}\lambda^{\unsat(\sigmahat)}$ is 
$q$ times the weight enumerator of the code~$C$.  Then, from (\ref{eq:PottsDef})
and~(\ref{eq:unsat}), letting $M$ be any generating matrix
for~$C$, 
$$
\ZPotts(G;q,\gamma)=q\lambda^{-(1-1/p)qm}
\,W_M(\lambda).
$$
To see where the factor~$q$ comes from, consider the 
assignments~$\sigmahat$ satisfying 
\begin{equation}\label{eq:qto1}
\sum_{i\in[k],v\in V}\sigmahat(x_i^v)\,\bfa_i^v=\bfb,
\end{equation}
for some $\bfb\in\Fp^{qm}$.
For every $i\in [k]$ and every edge 
$(u,v)\in E$,
there is an equation in $\Lambda$ 
specifying the value of $\sigmahat(x_i^v)-\sigmahat(x_i^u)$.  
Thus, since $G$ is connected, the vector $\bfb$ determines 
$\sigmahat$ once the partial assigment 
$(\sigmahat(x_1^r),\ldots,\sigmahat(x_k^r))$ is specified for some 
distinguished vertex $r\in V$.
Conversely, each of the $q$ partial assignments
$(\sigmahat(x_1^r),\ldots,\sigmahat(x_k^r))$ extends to a 
total assignment satisfying~(\ref{eq:qto1}).
\end{proof}

\bibliographystyle{plain}
\bibliography{mybibfile}

\end{document}